\documentclass{article}

\PassOptionsToPackage{dvipsnames}{xcolor}
\usepackage[preprint]{neurips_2026}

\usepackage[utf8]{inputenc}
\usepackage[T1]{fontenc}
\usepackage{hyperref}
\usepackage{url}
\usepackage{booktabs}
\usepackage{amsfonts}
\usepackage{amsmath}
\usepackage{amssymb}
\usepackage{subcaption}
\usepackage{nicefrac}
\usepackage{caption}
\usepackage{microtype}
\usepackage{graphicx}
\usepackage{algorithm}
\usepackage{algorithmic}
\usepackage{tcolorbox}
\usepackage{enumitem}
\usepackage{multirow}
\usepackage{amsthm}
\usepackage{xcolor}
\usepackage{colortbl}
\definecolor{darkcerulean}{rgb}{0.03, 0.27, 0.49}
\definecolor{burgundy}{rgb}{0.7, 0.0, 0.13}
\definecolor{majorelleblue}{rgb}{0.38, 0.31, 0.86}
\definecolor{forestgreen}{rgb}{0.13, 0.55, 0.13}
\hypersetup{
    colorlinks=true, % Set to 'true' if you want colored links
    citecolor=darkcerulean,   % Define the color of citation links
    linkcolor=burgundy,    % Define the color of internal links
    urlcolor=majorelleblue     % Define the color of external hyperlinks
}

\DeclareMathOperator*{\argmax}{arg\,max}

\newcounter{noteHYTctr} 

\usepackage{titlesec}
\titlespacing*{\section}{0pt}{0.1\baselineskip}{0.1\baselineskip}
\titlespacing*{\subsection}{0pt}{0.1\baselineskip}{0.1\baselineskip}
\titlespacing*{\subsubsection}{0pt}{0.1\baselineskip}{0.1\baselineskip}
\titlespacing*{\paragraph}{0pt}{0.4\baselineskip}{0.5em}

\AtBeginDocument{%
  \setlength{\textfloatsep}{3pt plus 1pt minus 1pt}%
  \setlength{\intextsep}{3pt plus 1pt minus 1pt}%
  \setlength{\floatsep}{3pt plus 1pt minus 1pt}%
  \setlength{\dbltextfloatsep}{3pt plus 1pt minus 1pt}%
  \setlength{\dblfloatsep}{3pt plus 1pt minus 1pt}%
  \setlength{\abovecaptionskip}{2pt}%
  \setlength{\belowcaptionskip}{0pt}%
}

\newtheorem{assumption}{Assumption}
\newtheorem{theorem}{Theorem}
\newtheorem{lemma}{Lemma}
\newtheorem{proposition}{Proposition}
\newtheorem{corollary}{Corollary}

\title{LARGER: Lexically Anchored Repository Graph Exploration and Retrieval}

\author{
Yuntong Hu$^{1,*}$ \quad Tongli Su$^{1,*}$ \quad Liang Zhao$^{1,2,\dagger}$ \quad Bowen Zhu$^{2}$ \quad Hasibul Haque$^{2}$ \\
\\
$^{1}$Emory University \quad $^{2}$CausalDynamics.com
}

\begin{document}
\maketitle
\renewcommand{\thefootnote}{}
\footnotetext{\hspace{-1.8em}$^{*}$Equal contribution. $^{\dagger}$Corresponding author: \texttt{liang.zhao@emory.edu}.}
\vspace{-15pt}
\begin{abstract}
Repository-level coding agents must first \emph{localize} the files and symbols relevant to a task; failures at this stage can cascade across downstream objectives ranging from patch generation to test writing and codebase question answering. Existing agents navigate repositories primarily through lexical search, often missing structural relations such as imports, call chains, type hierarchies, and code--test links. Graph-based retrieval can recover such dependencies, but existing approaches often require separate graph tools or traversal stages that fragment the agent's interaction loop. We formalize repository context localization as \emph{Lexically Anchored Structural Localization}, where success depends on turning lexical matches into high-precision structural entry points and expose the most useful confidence-filtered local neighborhoods within the agent's existing search loop. We introduce \textbf{LARGER} (\textbf{Lexically Anchored Repository Graph Exploration and Retrieval}), a lexically anchored active-set retrieval framework that starts from lexical matches, aligns them to graph anchors, and performs confidence-filtered local expansion within the agent's existing search loop. LARGER integrates directly into existing CLI coding agents without requiring external graph databases or specialized graph interfaces. Across four benchmarks spanning localization, test generation, and codebase understanding, LARGER improves file-level Acc@5 on LocBench by +13.9 points with tuned hyperparameters and still gains +11.8 points with fixed hyperparameters over the strongest baseline, while delivering consistent gains on MuLocBench, SWE-Atlas Test Writing, and SWE-Atlas Codebase QA.
\end{abstract}

\section{Introduction}

LLM-powered coding agents can now fix bugs, implement features, and resolve complex GitHub issues across large repositories~\citep{yang2024sweagent, xia2024agentless, wang2024openhands}. Yet the success depends on successful \emph{repository context localization}: given a natural-language issue, identifying the small set of files and symbols that must be inspected and potentially modified. When localization fails, downstream objectives can be greatly undermined. As repositories grow to thousands of files and hundreds of thousands of lines, effective localization becomes both more critical and more difficult. This challenge stems not only from repository size but from structural complexity: code repositories are not flat text corpora, but sparse, typed relational systems whose relevant evidence is often distributed across imports, call chains, inheritance hierarchies, test fixtures, and documentation.

Existing repository context localization systems rely on two retrieval paradigms. \textbf{Lexical retrieval} underlies modern CLI coding agents such as Claude Code and Codex: keyword-driven tools such as \texttt{grep} preserve fast adaptive exploration, but operate on a flat repository view and often miss structurally relevant but lexically mismatched evidence. \textbf{Structured retrieval} methods incorporate repository structure through graph traversal or external memory~\citep{chen2025locagent, wang2025improving}, improving recall by exploiting dependencies across artifacts, but often introduce substantial overhead and fragment the agent's interaction loop across separate graph tools, databases, or traversal stages.

Our goal is to combine the best of both paradigms: the adaptive, low-friction exploration of lexical retrieval and the dependency-aware recall of structured retrieval, without inheriting the flat-view failure mode of the former or the heavy, fragmented retrieval loop of the latter. We therefore formalize repository context localization as \emph{Lexically Anchored Structural Localization} and propose \textbf{Lexically Anchored Repository Graph Exploration and Retrieval (LARGER)}. In LARGER, the agent's own lexical search queries anchor into the repository graph, and confidence-scored local expansion surfaces structurally related evidence within the same search output, so graph evidence is delivered \emph{within} lexical observations rather than through a separate graph tool, database, or traversal loop. Across four benchmarks (LocBench~\citep{chen2025locagent}, MuLocBench~\citep{zhang2025mulocbench}, SWE-Atlas Test Writing, and SWE-Atlas Codebase QA~\citep{scale2026sweatlas}), LARGER improves over strong baselines including BM25~\citep{robertson2009probabilistic}, Agentless~\citep{xia2024agentless}, LocAgent~\citep{chen2025locagent}, SWE-agent~\citep{yang2024sweagent}, OpenHands~\citep{wang2024openhands}, Claude Code, and Codex by up to $+13.9$ points on Acc@5 and $+12.2$ points on Recall@5, and surpasses strong proprietary baselines on SWE-Atlas Codebase QnA and Test Writing (\hyperref[sec:experiments]{Section~\ref*{sec:experiments}}).

The main contributions of this paper are as follows:
\begin{enumerate}[leftmargin=*, itemsep=2pt]
    \item \textbf{Formulation.} We formalize repository context localization as \emph{Lexically Anchored Structural Localization}, where a fixed agent policy localizes by turning lexical observations into structural anchors and progressively exposing compact local neighborhoods within the agent's existing search loop, without introducing new graph-specific actions (\hyperref[sec:formulation]{Section~\ref*{sec:formulation}}).
    \item \textbf{Method.} We decompose this joint objective into two tractable subproblems, \emph{graph quality} and \emph{retrieval efficiency}, and propose a divide-and-conquer approximation that reduces global subgraph search to lexical anchoring plus confidence-scored local expansion (\hyperref[sec:LARGER]{Section~\ref*{sec:LARGER}}).
    \item \textbf{System.} We present a practical architecture realizing both objectives: graph quality through multi-language AST-based construction, edge weighting and community detection; retrieval efficiency through lightweight sidecar storage, dynamic lexical propagation, and budget-aware filtering, all integrated into the CLI coding agent's existing search loop without additional tools (\hyperref[sec:LARGER]{Section~\ref*{sec:LARGER}}).
\end{enumerate}

\section{Related Work}

\paragraph{CLI Coding Agents.}
Modern CLI coding agents, spanning research systems such as SWE-agent~\citep{yang2024sweagent} and OpenHands~\citep{wang2024openhands} as well as widely deployed products including Claude Code, Codex, and OpenCode, have become a dominant interface for repository-level software engineering. They expose the LLM to the codebase through a small, fixed action set centered on shell-level tools, a minimalist design that pairs naturally with the identifier-centric nature of code and underpins strong performance on various coding tasks. However, this lexical-only interface leaves structural dependencies invisible unless surfaced through additional reasoning steps.

\paragraph{Repository-Level Code Localization.}
Agentic systems such as SWE-agent~\citep{yang2024sweagent}, OpenHands~\citep{wang2024openhands}, and OrcaLoca~\citep{yu2025orcaloca} localize via iterative shell-level search, while non-agentic pipelines such as Agentless~\citep{xia2024agentless}, CodePlan~\citep{bairi2024codeplan}, and HiLoRM~\citep{zhang2025hierarchical} use static multi-stage retrieval; RepoNavigator~\citep{zhang2026reponavigator} and SweRank+~\citep{reddy2025swerankplus} further learn exploration and ranking policies. These methods remain predominantly lexical, motivating work that augments localization with explicit graph retrieval over the repository.

\paragraph{Graph Retrieval for Code.}
Graph retrieval~\citep{peng2024graphragsurvey, hu2024grag, he2024gretriever} is most effective for evidence that is sparse, multi-hop, and relationally typed~\citep{xiang2025whengraphs}, the regime of code repositories. Code-focused approaches~\citep{liu2026a2rag, yang2025graphsearch} differ in how graph evidence reaches the agent: \emph{static pipelines}~\citep{shi2024repograph, tao2025cgm} consume offline-built graphs but lack adaptive exploration; \emph{graph-tool interfaces}~\citep{liu2024codexgraph, shah2025ranger, vogel2026codebasememory} expose graph queries as new agent actions but fragment the search loop; \emph{graph-guided agents}~\citep{chen2025locagent, jiang2025cosil, li2025graphcodeagent, wang2025improving} steer exploration via graph structure but tightly couple traversal with agent control. LARGER instead delivers graph evidence through the agent's existing lexical channel, jointly addressing graph quality and retrieval efficiency by design without breaking the CLI agent regime.

\section{Problem Formulation}
\label{sec:formulation}

\paragraph{Repository Graph.}
\label{def:repo-graph}
A code repository is represented as a heterogeneous attributed graph
\begin{equation}
G = (V, E, \mathcal{X}),
\end{equation}
where the node set is a typed union
\begin{equation}
V = V^{\mathrm{dir}} \cup V^{\mathrm{file}} \cup V^{\mathrm{class}} \cup V^{\mathrm{func}},
\end{equation}
with $V^{\mathrm{dir}}$, $V^{\mathrm{file}}$, $V^{\mathrm{class}}$, and $V^{\mathrm{func}}$ denoting the sets of directory, file, class, and function nodes, respectively. Let $\mathcal{R}$ denote the finite set of relation types, and let $\mathcal{T}$ denote the space of text strings. The edge set $E \subseteq V \times \mathcal{R} \times V$ consists of typed directed edges $(u,r,v)$ with $u,v \in V$ and $r \in \mathcal{R}$, and the attribute map $\mathcal{X}: V \to \mathcal{T}$ assigns each node a textual attribute such as a file path, code content, or documentation snippet. At snapshot $c$, the repository graph is
\[
G_c = (V_c, E_c, \mathcal{X}_c),
\]
where $V_c \subseteq V$ is the set of nodes present at $c$, $E_c \subseteq E$ contains the edges induced by $V_c$, and $\mathcal{X}_c = \mathcal{X}|_{V_c}$.

The relation set $\mathcal{R}$ captures structural and cross-artifact dependencies within the repository, such as \texttt{contains} (hierarchical structure), \texttt{imports} (module dependencies), \texttt{invokes} (function calls), and cross-artifact links (e.g., source-to-test or source-to-documentation connections). 
This graph representation provides a general abstraction of repository structure, capturing dependencies across files, functions, and auxiliary artifacts such as tests and documentation. 
Such a representation is broadly useful for tasks requiring structured reasoning over codebases.

\paragraph{Coding Agent.}
A modern coding agent is modeled as a policy $\pi$ implemented by a pretrained language model that interacts with a repository through a set of actions $\mathcal{U}$ (e.g., search, file inspection) and receives observations from the environment. 

The agent operates iteratively, producing a sequence of actions and observations while updating its internal context.
For such agents, code localization is a fundamental prerequisite for downstream tasks: if the relevant files, classes, or functions are not identified, patch generation, test writing, and code understanding all degrade.

\paragraph{Code Localization.}
\label{def:localization}
Given a natural-language issue description $q$ on a repository represented by the graph $G$, code localization aims to retrieve the task-relevant target set of code entities $Y(q, G) \subseteq V$. Letting $U(\mathcal{S} \mid q, G)$ denote a utility function that captures both the relevance of $\mathcal{S}$ and its retrieval cost, we model the target as
\begin{equation}
\label{eq:objective}
Y(q, G) = \mathcal{S}^*(q, G) = \argmax_{\mathcal{S} \subseteq V} U(\mathcal{S} \mid q, G).
\end{equation}
At inference time the agent operates on a specific snapshot $G_c$, so each instance of the localization problem reduces to evaluating \hyperref[eq:objective]{Eq.~\ref*{eq:objective}} at $G\!=\!G_c$.

Existing approaches fall into two paradigms:
\begin{itemize}[left=0pt]
    \item \textbf{Lexical retrieval.}
    Modern CLI coding agents (e.g., Claude Code, Codex, OpenCode) rely on keyword-driven tools such as \texttt{grep}. 
    While efficient and well aligned with identifier-centric codebases, lexical retrieval operates on a flat repository view and often fails to recover structurally relevant but lexically mismatched evidence.
    
    \item \textbf{Structured retrieval.}
    Recent methods incorporate repository structure through graph traversal or external memory~\citep{chen2025locagent, wang2025improving}, improving recall by exploiting dependencies across artifacts. 
    However, these approaches typically rely on heavyweight graph access and global traversal strategies that introduce substantial overhead and disrupt the agent’s interaction loop.
\end{itemize}

Lexical methods are thus efficient but incomplete, while structured methods are expressive but inefficient. Effective code localization must exploit structural dependencies without violating the agent's stepwise interaction constraints.

These requirements lead to several key challenges:
\begin{itemize}[left=0pt]
    \item \textbf{Structural accessibility.} How to represent repository structure so that task-relevant dependencies are reachable through compact local neighborhoods.
    \item \textbf{Efficient exposure.} How to expose useful structural evidence without incurring the overhead of global graph traversal or excessive context growth.
    \item \textbf{Sequential adaptation.} How to adapt retrieval decisions as the agent’s context evolves during interaction.
\end{itemize}

We address these challenges by coupling lexical anchoring with localized graph expansion, enabling efficient exploration while progressively incorporating structural evidence.

\section{LARGER: Lexically Anchored Repository Graph Exploration and Retrieval}
\label{sec:LARGER}

LARGER solves the snapshot specialization \eqref{eq:global-subgraph} of the localization objective \eqref{eq:objective} through a two-level decomposition.
At a fixed repository snapshot \(G_c\), it approximates the intractable subgraph-selection problem by an active-set loop that sequentially expands from lexical anchors (\hyperref[sec:inner-loop]{\S\ref*{sec:inner-loop}}) and solves each step by confidence-filtered local retrieval (\hyperref[sec:per-step]{\S\ref*{sec:per-step}}).
Across snapshots, it maintains the evolving repository graph through commit-aware alignment and summarizes the full procedure in \hyperref[alg:larger]{Algorithm~\ref*{alg:larger}} (\hyperref[sec:outer-loop]{\S\ref*{sec:outer-loop}}).

\subsection{Active-Set Loop}
\label{sec:inner-loop}
Fix a commit $c$ and a query $q$. Specialized to $G_c$, the localization objective in \hyperref[eq:objective]{Eq.~\ref*{eq:objective}} reduces to a subgraph-selection problem,
\begin{equation}
\hat g^{*} \;=\; \argmax_{\hat g \subseteq G_c} U(\hat g \mid q, G_c).
\label{eq:global-subgraph}
\end{equation}
Direct optimization of \hyperref[eq:global-subgraph]{Eq.~\ref*{eq:global-subgraph}} is intractable since the candidate space has cardinality $2^{|V_c|}$, and in the agentic setting evidence is exposed sequentially under bounded context and tool budgets rather than chosen globally in one shot. We therefore decompose discovery over $T$ interaction steps and approximate $\hat g^{*}$ by the cumulative subgraph that the agent has exposed by step $T$:
\begin{equation}
\hat g_{1:T} \;=\; \hat g_{1:T-1} \cup \hat g_T \;=\; \bigcup_{t=1}^{T} \hat g_t, \qquad \hat g_{1:0} = \emptyset,
\label{eq:cumulative}
\end{equation}
where $\hat g_t$ is the graph exposed to the agent at step $t$, and $\hat g_{1:T}$ is the cumulative discovered subgraph after $T$ steps, taken as our approximation $\hat g\!\approx\!\hat g_{1:T}$. The inner loop's design reduces to specifying the per-step expansion rule that produces $\hat g_t$.

Directly solving \hyperref[eq:global-subgraph]{Eq.~\ref*{eq:global-subgraph}} requires searching over exponentially many subgraphs of \(G_c\), while the agent can only observe a bounded amount of context at each step.
This mismatch motivates us to develop an active-set procedure: at step \(t\), we maintain a
small working set \(M_t \subseteq V_c\) of currently plausible anchor nodes and solve
only the restricted local problem around them. The active set is induced by the
agent's lexical observation,
\begin{equation}
M_t = \operatorname{align}\!\left(\Omega_{\mathrm{lex}}(a_t)\right),
\qquad
a_t = \pi(q, \mathcal{C}_{t-1}),
\label{eq:active}
\end{equation}
where \(\mathcal{C}_{t-1}\) is the accumulated context before step \(t\). Thus, \(M_t\)
contains only repository nodes aligned to the agent's current search result, avoiding
blind traversal over the full graph.
Given \(M_t\), the step-\(t\) restricted solve exposes confidence-filtered local
neighborhoods around active anchors:
\begin{equation}
\hat g_t =
\bigcup_{v \in M_t}
G_c\!\left[\mathcal N^{*}_{K,t}(v)\right],
\qquad
\mathcal N^{*}_{K,t}(v) \subseteq \mathcal N_K(v),
\label{eq:anchor-subgraph}
\end{equation}
where $\mathcal{N}_K(v) = \{u \in V_c : d_{G_c}(u,v) \le K\}$ is the $K$-hop neighborhood of $v$ in $G_c$, $\mathcal{N}_{K,t}^{*}(v)$ is the subset of those neighbors retained at step $t$, and $G_c[S]$ denotes the subgraph of $G_c$ induced by $S$. Substituting \hyperref[eq:anchor-subgraph]{Eq.~\ref*{eq:anchor-subgraph}} into \hyperref[eq:global-subgraph]{Eqs.~\ref*{eq:global-subgraph}--\ref*{eq:cumulative}} reduces the global subgraph search to a sequence of local selections parameterized by the active sets $\{M_t\}_{t=1}^{T}$ and the retained neighbors $\{\mathcal{N}_{K,t}^{*}(v)\}_{v \in M_t}$. The
exposed anchors and graph evidence are then folded into the bounded context,
\begin{equation}
\mathcal{C}_t =
\Pi\!\left(\mathcal{C}_{t-1} \cup M_t \cup \Gamma(M_t,q,\mathcal{C}_{t-1})\right),
\label{eq:context}
\end{equation}
which conditions the next action \(a_{t+1}=\pi(q,\mathcal{C}_t)\) and hence reselects the next
active set \(M_{t+1}\). The procedure therefore follows the active-set
pattern: restrict the global problem to a small active region, solve the restricted
problem, update the working context, and reselect the active set for the next step. After \(T\) active-set updates, the agent produces the final localization
\(R_{c,T}=\psi(\mathcal{C}_T,q)\), where \(\psi\) selects from the bounded accumulated
context \(\mathcal{C}_T\) rather than the full repository graph.

\subsection{Per-Step Local Solver}
\label{sec:per-step}
Given the active set \(M_t\), we next specify how the restricted step in \hyperref[eq:anchor-subgraph]{Eq.~\ref*{eq:anchor-subgraph}}
is solved efficiently.
Within step $t$, the joint utility of the per-anchor neighborhoods admits an additive upper bound, with equality when the neighborhoods are pairwise disjoint:
\begin{equation}
U\!\Bigl(\textstyle\bigcup_{v \in M_t} G_c\!\bigl[\mathcal{N}_{K,t}^{*}(v)\bigr] \,\Big|\, q, G_c\Bigr)
\;\le\;
\sum_{v \in M_t}
U\!\bigl(G_c\!\bigl[\mathcal{N}_{K,t}^{*}(v)\bigr] \,\big|\, q, G_c\bigr),
\label{eq:additive}
\end{equation}
which holds whenever $U$ is monotone submodular in the exposed set, a standard property of recall- and coverage-based utilities. In real repositories, lexical anchors typically target distinct roles or modules, so their $K$-hop neighborhoods overlap weakly and \hyperref[eq:additive]{Eq.~\ref*{eq:additive}} is close to tight. LARGER therefore optimizes the additive surrogate, decomposing the step-$t$ problem into independent per-anchor subproblems,
\begin{equation}
\mathcal{N}_{K,t}^{*}(v) \;=\;
\argmax\nolimits_{\mathcal{N} \subseteq \mathcal{N}_K(v)}
U\!\bigl(G_c[\mathcal{N}] \mid q, G_c\bigr),
\qquad \forall\, v \in M_t.
\label{eq:local-opt}
\end{equation}
\hyperref[eq:local-opt]{Eq.~\ref*{eq:local-opt}} is small and tractable for fixed $M_t$. The remaining design questions are (i) how the active set $M_t$ is initialized and updated, and (ii) how the local utility in \hyperref[eq:local-opt]{Eq.~\ref*{eq:local-opt}} is approximated in practice, which are addressed as follows:

\paragraph{Lexical seed identification to address Question (i)}
At step $t$, given the query $q$ and the accumulated context $\mathcal{C}_{t-1}$ defined below, the agent issues a lexical search action. The action returns lexical matches $O_t^{\mathrm{lex}} = \Omega_{\mathrm{lex}}(a_t)$, which align to a node set
\begin{equation}
M_t \;=\; \mathrm{align}\!\bigl(O_t^{\mathrm{lex}}\bigr) \;\subseteq\; V_c.
\label{eq:anchor}
\end{equation}
$M_t$ is the active set at step $t$: high-precision entry points whose discovery cost is amortized into the agent's existing search action and incurs no additional tool call.

\paragraph{Graph-conditioned local retrieval to address Question (ii).}
For each anchor $v \in M_t$, we approximate the local utility in \hyperref[eq:local-opt]{Eq.~\ref*{eq:local-opt}} by an additive decomposition of $U$ over the exposed nodes. Adopting the convention $U(\emptyset \mid q, \mathcal{C}_{t-1}) = 0$ and writing $\mathrm{score}_t(u \mid v, q, \mathcal{C}_{t-1}) \;=\; U\!\bigl(G_c[\{u\}] \mid q, \mathcal{C}_{t-1}\bigr)$ for the singleton utility of exposing $u$ given the current context, the joint utility decomposes as
\begin{equation}
U\!\bigl(G_c[\mathcal{N}] \mid q, \mathcal{C}_{t-1}\bigr)
\;\approx\;
\sum\nolimits_{u \in \mathcal{N}} \mathrm{score}_t(u \mid v, q, \mathcal{C}_{t-1}).
\label{eq:scoring}
\end{equation}
The decomposition is exact when $U$ is additive in the exposed nodes (e.g., recall and coverage utilities, which sum hits over $Y$) and tight up to a low-overlap term when $U$ is monotone submodular. Substituting \hyperref[eq:scoring]{Eq.~\ref*{eq:scoring}} into the per-anchor problem of \hyperref[eq:local-opt]{Eq.~\ref*{eq:local-opt}}, together with a cardinality budget $|\mathcal{N}| \le k$ and a confidence filter $\omega(v,u) \ge \theta$, yields a closed-form selection:
\begin{equation}
\mathcal{N}_{K,t}^{*}(v) \;=\;
\argmax_{\substack{\mathcal{N}\,\subseteq\,\mathcal{N}_K^{\theta}(v) \\ |\mathcal{N}|\,\le\,k}}
\;\sum_{u \in \mathcal{N}} \mathrm{score}_t(u \mid v, q, \mathcal{C}_{t-1}),
\qquad
\mathcal{N}_K^{\theta}(v) \;=\; \bigl\{u \in \mathcal{N}_K(v) : \omega(v,u) \ge \theta\bigr\},
\label{eq:topk}
\end{equation}
where $\omega(v,u) \in [0,1]$ is the provenance-aware edge confidence (\hyperref[app:graph-impl]{Appendix~\ref*{app:graph-impl}}); the argmax solves to the $k$ highest-scoring elements of $\mathcal{N}_K^{\theta}(v)$. Aggregating across anchors yields the step's graph-augmented evidence,
\begin{equation}
\Gamma(M_t, q, \mathcal{C}_{t-1}) \;=\; \bigcup_{v \in M_t} \mathcal{N}_{K,t}^{*}(v).
\label{eq:expand-aggregate}
\end{equation}
\hyperref[eq:scoring]{Eqs.~\ref*{eq:scoring}--\ref*{eq:topk}} make the per-anchor problem in \hyperref[eq:local-opt]{Eq.~\ref*{eq:local-opt}} executable in $O(|\mathcal{N}_K(v)|)$ time, independently of $|V_c|$.

\subsection{Theoretical Analysis}
We now summarize the consequences of the active-set solver for recall, step count,
and token cost.
\begin{theorem}[Recall dominance]
\label{thm:recall-dominance}
Under a shared policy $\pi$ and identical query stream, let $\mathcal{C}_t^{\mathrm{lex}}$ denote the context produced with $\Gamma \equiv \emptyset$ and $\mathcal{C}_t^{\mathrm{LARGER}}$ the context produced by the full system. Then for every step $t$,
\[
\mathcal{C}_t^{\mathrm{lex}} \;\subseteq\; \mathcal{C}_t^{\mathrm{LARGER}},
\qquad \text{and hence} \qquad
\mathrm{Recall}\bigl(\mathcal{C}_t^{\mathrm{LARGER}}, Y\bigr) \;\geq\; \mathrm{Recall}\bigl(\mathcal{C}_t^{\mathrm{lex}}, Y\bigr).
\]
\end{theorem}

A finer quantitative version characterizes the recall gap exactly.

\begin{theorem}[Quantitative discovery advantage]
\label{thm:quant-discovery}
For any step $T$, let $H_T = \bigl(\mathcal{C}_T^{\mathrm{LARGER}} \setminus \mathcal{C}_T^{\mathrm{lex}}\bigr) \cap Y$ denote the ground-truth nodes that LARGER has surfaced through graph expansion but the non-graph regime has not. Then
\[
\mathrm{Recall}\bigl(\mathcal{C}_T^{\mathrm{LARGER}}, Y\bigr) - \mathrm{Recall}\bigl(\mathcal{C}_T^{\mathrm{lex}}, Y\bigr) \;=\; |H_T|/|Y|.
\]
Whenever some relevant node is reachable from a realized lexical anchor within $K$ hops along edges of confidence $\geq \theta$ but does not itself appear as a lexical match, $H_T$ is non-empty and the gap is strict; the non-graph regime then requires additional steps to close it. The full structural condition is formalized as \hyperref[ass:reachability]{Assumption~\ref*{ass:reachability}} in \hyperref[app:theory]{Appendix~\ref*{app:theory}}.
\end{theorem}

\noindent The proofs of Theorems~\ref{thm:recall-dominance} and~\ref{thm:quant-discovery} are deferred to \hyperref[app:theory]{Appendix~\ref*{app:theory}}.

If the lexical matcher returns at most $m$ anchors per step and each rendered graph node contributes at most $L_{\mathrm{node}}^{\max}$ prompt tokens, then $|\Gamma(M_t, q, \mathcal{C}_{t-1})| \leq mk$ and $\Gamma$ injects at most $\Delta = m\,k\,L_{\mathrm{node}}^{\max}$ tokens per step beyond lexical retrieval, independently of $|V_c|$ (\hyperref[lem:bounded-aug]{Lemma~\ref*{lem:bounded-aug}}, \hyperref[app:theory]{Appendix~\ref*{app:theory}}). Combined with recall dominance, this yields:

\begin{proposition}[Step-count dominance]
\label{prop:step-dominance}
Fix a policy $\pi$ and a query stream $\{q_t\}_{t \geq 1}$ shared by both regimes. For any recall threshold $R \in (0,1]$, define
\[
T^{(\cdot)}(R) \;=\; \min\bigl\{\,t \,:\, \mathrm{Recall}\bigl(\mathcal{C}_t^{(\cdot)}, Y\bigr) \geq R\,\bigr\}, \quad T^{(\cdot)}(R) = \infty \text{ if never reached.}
\]
Then $T^{\mathrm{LARGER}}(R) \;\leq\; T^{\mathrm{lex}}(R)$ for all $R \in (0,1]$.
\end{proposition}

This is a conservative bound: in practice, the richer observations from LARGER also lead the agent to issue more targeted subsequent queries, further amplifying step savings.

\begin{corollary}[Conditional token-cost dominance]
\label{cor:token-cost}
Let $C_{\mathrm{step}}^{\mathrm{lex}}$ denote the per-step token cost of the non-graph regime and $\Delta = m\,k\,L_{\mathrm{node}}^{\max}$ the per-step graph-augmentation overhead. For any target recall $R^* \in (0,1]$,
\[
\mathrm{Cost}_{\mathrm{LARGER}}(R^*) \;\leq\; T^{\mathrm{LARGER}}(R^*)\,\bigl(C_{\mathrm{step}}^{\mathrm{lex}} + \Delta\bigr),
\]
and LARGER strictly dominates the non-graph regime in total token cost whenever
\[
\frac{T^{\mathrm{lex}}(R^*) - T^{\mathrm{LARGER}}(R^*)}{T^{\mathrm{lex}}(R^*)}
\;>\;
\frac{\Delta}{C_{\mathrm{step}}^{\mathrm{lex}} + \Delta}.
\]
\end{corollary}

\noindent Proofs and supporting lemmas are deferred to \hyperref[app:theory]{Appendix~\ref*{app:theory}}. In real-world code-repository settings, $\Delta \ll C_{\mathrm{step}}^{\mathrm{lex}}$ (a few KB of injected neighbors), so even modest step savings yield a strict cost win. The gain thus comes from fewer agent interactions, not from cheaper individual steps.

\begin{algorithm}[t]\small
\caption{\small LARGER: Commit-Aware Active-Set Solver for \hyperref[eq:objective]{Eq.~\ref*{eq:objective}}}
\label{alg:larger}
\begin{algorithmic}[1]\small
\REQUIRE Target commits \(\{c\}\) with input query \(q_c\) at each; initial graph \(G_0\); hyperparameters \(T,K,k,\theta\)
\ENSURE Localizations \(\{R_{c,T}\}\)

\FOR{each target commit \(c\)}
    \STATE $G_c \gets \mathcal{A}\!\bigl(G_{c-1},\,\mathrm{diff}(c{-}1,\,c)\bigr)$ by \hyperref[eq:alignment]{Eq.~\ref*{eq:alignment}}.
    \STATE Reduce \hyperref[eq:objective]{Eq.~\ref*{eq:objective}} to the fixed-snapshot objective \hyperref[eq:global-subgraph]{Eq.~\ref*{eq:global-subgraph}} for query \(q_c\) on \(G_c\); initialize \(\mathcal{C}_0=\emptyset\), \(\hat g_{1:0}=\emptyset\).
    \FOR{\(t=1,\ldots,T\)}
        \STATE Select active anchors \(M_t\) by \hyperref[eq:active]{Eq.~\ref*{eq:active}}.
        \STATE Select graph-augmented evidence by \hyperref[eq:scoring]{Eqs.~\ref*{eq:scoring}--\ref*{eq:expand-aggregate}}.
        \STATE Accumulate the restricted subgraph by \hyperref[eq:cumulative]{Eqs.~\ref*{eq:cumulative} and~\ref*{eq:anchor-subgraph}}.
        \STATE Update bounded context by \hyperref[eq:context]{Eq.~\ref*{eq:context}}.
    \ENDFOR
    \STATE Output \(R_{c,T}=\psi(\mathcal{C}_T,q_c)\).
\ENDFOR
\end{algorithmic}
\end{algorithm}
\normalsize
\subsection{Commit-Aware Maintenance and Full Algorithm}
\label{sec:outer-loop}

The above solver for a snapshot $G_c$ at commit $c$ will be sequentialized to dynamic cases with sequential commits. To avoid reconstructing the graph from scratch at every commit, which can be prohibitively expensive (as stated in \hyperref[app:alignment]{Appendix~\ref*{app:alignment}}), we propose a lazy, local update. Specifically, let $\mathrm{diff}(c{-}1, c)$ denote the set of files added or modified between commits $c{-}1$ and $c$. LARGER defines $G_c$ from $G_{c-1}$ by an \emph{alignment} operator $\mathcal{A}$ that drops nodes whose files are no longer present at $c$ and merges in the parse output of the changed files:
\begin{equation}
G_c \;=\; \mathcal{A}\!\bigl(G_{c-1},\, \mathrm{diff}(c{-}1, c)\bigr).
\label{eq:alignment}
\end{equation}
The cached graph $G_{c-1}$ is never mutated; the unchanged subgraph carries over to $G_c$ together with its edge confidence weights $\omega$ and community labels $\kappa$ (\hyperref[app:graph-impl]{Appendix~\ref*{app:graph-impl}}), and only the changed files are re-parsed and re-linked.

Altogether, our LARGER algorithm is summarized in \hyperref[alg:larger]{Algorithm~\ref*{alg:larger}}, which nests two loops: a commit-switching outer loop and the agent-step inner loop of \hyperref[sec:inner-loop]{\S\ref*{sec:inner-loop}}; per-commit queries enter as inputs. At each target commit $c$ (line~1), the \emph{commit-switching loop} first refreshes the snapshot graph by the alignment update (line~2), so that $G_c$ is up-to-date without paying full reconstruction cost. After reducing to its fixed-snapshot form on $G_c$ for the input query $q_c$ and initializing the empty context (line~3), the \emph{agent-step loop} (lines~4--9) iterates for $T$ steps: at each step $t$ it (i)~refreshes the active anchor set $M_t$ from current lexical matches (line~5, \hyperref[eq:active]{Eq.~\ref*{eq:active}}), (ii)~assembles graph-augmented evidence by retrieving $k$-hop neighbors that pass the confidence threshold $\theta$ and re-weighting them with the community prior (line~6, \hyperref[eq:scoring]{Eqs.~\ref*{eq:scoring}--\ref*{eq:expand-aggregate}}), (iii)~accumulates these into the restricted subgraph $\hat{g}_t$ (line~7, \hyperref[eq:cumulative]{Eqs.~\ref*{eq:cumulative} and~\ref*{eq:anchor-subgraph}}), and (iv)~projects the result into the bounded context $\mathcal{C}_t$ under the per-step token budget (line~8, \hyperref[eq:context]{Eq.~\ref*{eq:context}}). The readout $\psi$ then outputs the localization $R_{c,T}=\psi(\mathcal{C}_T,q_c)$ (line~10) before the outer loop advances to the next commit.

\section{Experiments}
\label{sec:experiments}

Our experiments evaluate whether jointly improving graph quality and retrieval efficiency improves localization and downstream agent performance. We report three views of the evidence: (1)~main results against lexical, procedure-based, agent-based, and prior graph-enhanced baselines (\hyperref[sec:main_results]{\S\ref*{sec:main_results}}); (2)~computational efficiency in wall-clock time, tokens, and cost (\hyperref[sec:efficiency]{\S\ref*{sec:efficiency}}); and (3)~component ablations isolating the contribution of graph expansion, confidence scoring, and community priors (\hyperref[sec:ablation]{\S\ref*{sec:ablation}}).

\subsection{Experimental Setup}
\label{sec:setup}

\noindent\textbf{Datasets.}
We evaluate on four benchmarks (\hyperref[tab:dataset_summary]{Table~\ref*{tab:dataset_summary}}, \hyperref[app:datasets]{Appendix~\ref*{app:datasets}}). For \emph{code localization}, we use \textbf{LocBench}~\citep{chen2025locagent} (560 issues, 5 repositories) and \textbf{MuLocBench}~\citep{zhang2025mulocbench} (1,100 multi-file issues, 46 repositories). For \emph{downstream tasks} on 11 multilingual repositories (Python, Go, TypeScript, C), we use \textbf{SWE-Atlas Test Writing} (90 test-generation tasks) and \textbf{SWE-Atlas Codebase QA} (124 code-understanding questions)~\citep{scale2026sweatlas}. The SWE-Atlas tasks test whether retrieval gains transfer beyond localization.

\noindent\textbf{Metrics.}
For localization, we report file-level Acc@$K$ and Recall@$K$ in the main paper; full baseline results, including function-level metrics and additional rankings (Hit@$K$, MAP, MRR), are deferred to \hyperref[app:full_results]{Appendix~\ref*{app:full_results}}. For Test Writing, we combine Docker-based execution with mutation testing, rubric-based LLM judging of quality and coverage, and manifest correctness verification. For Codebase QA, we use rubric-based LLM judging against must-have criteria. We additionally report wall-clock time, token consumption, and cost. We report the average results of three independent runs of each method.

\noindent\textbf{Baselines.}
We compare against four categories: \emph{lexical} (\textbf{BM25}~\citep{robertson2009probabilistic}); \emph{procedure-based} (\textbf{Agentless}~\citep{xia2024agentless}); \emph{agent-based, non-graph} (\textbf{SWE-agent}~\citep{yang2024sweagent}, \textbf{OpenHands}~\citep{wang2024openhands}, \textbf{Codex}, \textbf{Claude Code}, and \textbf{mini-swe-agent}); and \emph{agent-based, graph-enhanced} (\textbf{CoSIL}~\citep{jiang2025cosil} and \textbf{LocAgent}~\citep{chen2025locagent}).

\noindent\textbf{Implementation.}
All in-house LLM-based methods use GPT-5.2. LARGER uses default hyperparameters $k=10$ and $\theta=0.5$; the sidecar graph index is built once per repository and reused across issues, and runtime augmentation occurs inside the agent's existing search loop with no additional tool calls. \hyperref[tab:main_results]{Table~\ref*{tab:main_results}} reports both the default fixed setting and an accuracy-oriented tuned setting for LARGER. All agent-based methods use matched exploration-step and context-window budgets. Full setup details and hyperparameter sweeps are provided in \hyperref[app:datasets]{Appendix~\ref*{app:datasets}}.

\subsection{Main Observations}
\label{sec:main_results}

We organize the empirical findings into a single \emph{Main Observations} subsection that spans both accuracy and efficiency. \hyperref[tab:main_results]{Table~\ref*{tab:main_results}} reports the file-level localization results on MuLocBench and LocBench.

\noindent\begin{minipage}{\textwidth}
\centering
\small
\resizebox{\textwidth}{!}{
\begin{tabular}{l cccc cccc}
\toprule
\multirow{2}{*}{\textbf{Method}} & \multicolumn{4}{c}{\textbf{LocBench}} & \multicolumn{4}{c}{\textbf{MuLocBench}} \\
\cmidrule(lr){2-5} \cmidrule(lr){6-9}
 & Acc@1 & Acc@5 & Recall@1 & Recall@5 & Acc@1 & Acc@5 & Recall@1 & Recall@5 \\
\midrule
\multicolumn{9}{l}{\textit{Lexical retrieval}} \\
BM25 & 23.4 & 49.3 & 26.9 & 55.7 & 10.6 & 24.5 & 16.4 & 35.3 \\
\midrule
\multicolumn{9}{l}{\textit{Procedure-based}} \\
Agentless & 56.1 & 68.9 & 61.2 & 76.6 & 20.4 & 28.9 & 29.5 & 42.0 \\
\midrule
\multicolumn{9}{l}{\textit{Agent-based}} \\
SWE-agent & 60.0 & 70.9 & 65.8 & 79.3 & 21.0 & 31.3 & 30.4 & 44.5 \\
OpenHands & 57.7 & 67.5 & 63.2 & 75.3 & 27.3 & 43.7 & 38.0 & 57.6 \\
Codex & 63.2 & 74.1 & 69.3 & 82.8 & 28.0 & 50.0 & 39.8 & 65.1 \\
\rowcolor{black!8} Claude Code$^{*}$ & 65.1 & 75.2 & 71.8 & 85.6 & \textbf{29.2} & 54.9 & \textbf{41.4} & 69.0 \\
mini-swe-agent & 60.6 & 70.7 & 66.9 & 79.7 & 26.9 & 46.6 & 38.1 & 61.4 \\
\midrule
\multicolumn{9}{l}{\textit{Agent-based (graph-enhanced)}} \\
CoSIL & 54.6 & 68.2 & 59.7 & 75.7 & 18.1 & 27.0 & 27.0 & 40.1 \\
LocAgent & 56.0 & 65.3 & 62.0 & 73.8 & 17.9 & 25.7 & 27.3 & 39.0 \\
\midrule
\rowcolor{cyan!8} \textbf{LARGER$_{\text{Fixed}}$ (Ours)} & 72.9 & 87.0 & 77.4 & 90.1 & 27.4 & 55.7 & 39.4 & 68.6 \\
\quad $\Delta_{\text{vs.\ Codex}}$ & {\color{forestgreen}+9.7} & {\color{forestgreen}+12.9} & {\color{forestgreen}+8.1} & {\color{forestgreen}+7.3} & {\color{burgundy}-0.6} & {\color{forestgreen}+5.7} & {\color{burgundy}-0.4} & {\color{forestgreen}+3.5} \\
\quad $\Delta_{\text{vs.\ Claude Code}^{*}}$ & {\color{forestgreen}+7.8} & {\color{forestgreen}+11.8} & {\color{forestgreen}+5.6} & {\color{forestgreen}+4.5} & {\color{burgundy}-1.8} & {\color{forestgreen}+0.8} & {\color{burgundy}-2.0} & {\color{burgundy}-0.4} \\
\rowcolor{cyan!8}
\textbf{LARGER$_{\text{Tuned}}$ (Ours)} & \textbf{77.1} & \textbf{89.1} & \textbf{81.5} & \textbf{92.3} & 28.0 & \textbf{60.0} & 40.0 & \textbf{72.0} \\
\quad $\Delta_{\text{vs.\ Codex}}$ & {\color{forestgreen}+13.9} & {\color{forestgreen}+15.0} & {\color{forestgreen}+12.2} & {\color{forestgreen}+9.5} & 0.0 & {\color{forestgreen}+10.0} & {\color{forestgreen}+0.2} & {\color{forestgreen}+6.9} \\
\quad $\Delta_{\text{vs.\ Claude Code}^{*}}$ & {\color{forestgreen}+12.0} & {\color{forestgreen}+13.9} & {\color{forestgreen}+9.7} & {\color{forestgreen}+6.7} & {\color{burgundy}-1.2} & {\color{forestgreen}+5.1} & {\color{burgundy}-1.4} & {\color{forestgreen}+3.0} \\
\bottomrule
\end{tabular}
}
\captionsetup{font=small}
\captionof{table}{File-level localization results (\%). Acc@$K$: fraction of instances where \emph{all} ground-truth files appear in the top-$K$ predictions. Recall@$K$: fraction of ground-truth files recovered within the top-$K$ predictions. All LLM-based methods use GPT-5.2, except Claude Code$^{*}$ (light gray row), which uses Claude-Opus-4.6 and is therefore not directly comparable in absolute terms. $\Delta$ rows report absolute point differences against Codex (the strongest same-backbone baseline) and Claude Code$^{*}$ (different backbone); green/red indicate gains/deficits. \textbf{Bold} indicates the best value in each column.}
\label{tab:main_results}
\end{minipage}

\paragraph{Best Acc@5 and Recall@5 on both benchmarks.}
LARGER achieves the best Acc@5 on both LocBench and MuLocBench under the fixed setting. On LocBench, $\text{LARGER}_{\text{Fixed}}$ pushes Acc@5 from the strongest-baseline value of 75.2 to 87.0 ($+11.8$ points) and Recall@5 above 90; on MuLocBench it improves Acc@5 from 50.0 to 55.7 ($+5.7$ points) and Recall@5 from 65.1 to 68.6. The per-repository oracle setting ($\text{LARGER}_{\text{Tuned}}$) further reaches 89.1 on LocBench and 60.0 on MuLocBench. The main exception is MuLocBench Acc@1, where Claude Code still leads; LARGER's advantage lies less in ranking the single best file first and more in recovering a broader set of relevant files within a short candidate list.

\begin{figure*}[h]
  \centering
  \begin{subfigure}[t]{0.325\textwidth}
    \centering
    \includegraphics[width=\linewidth]{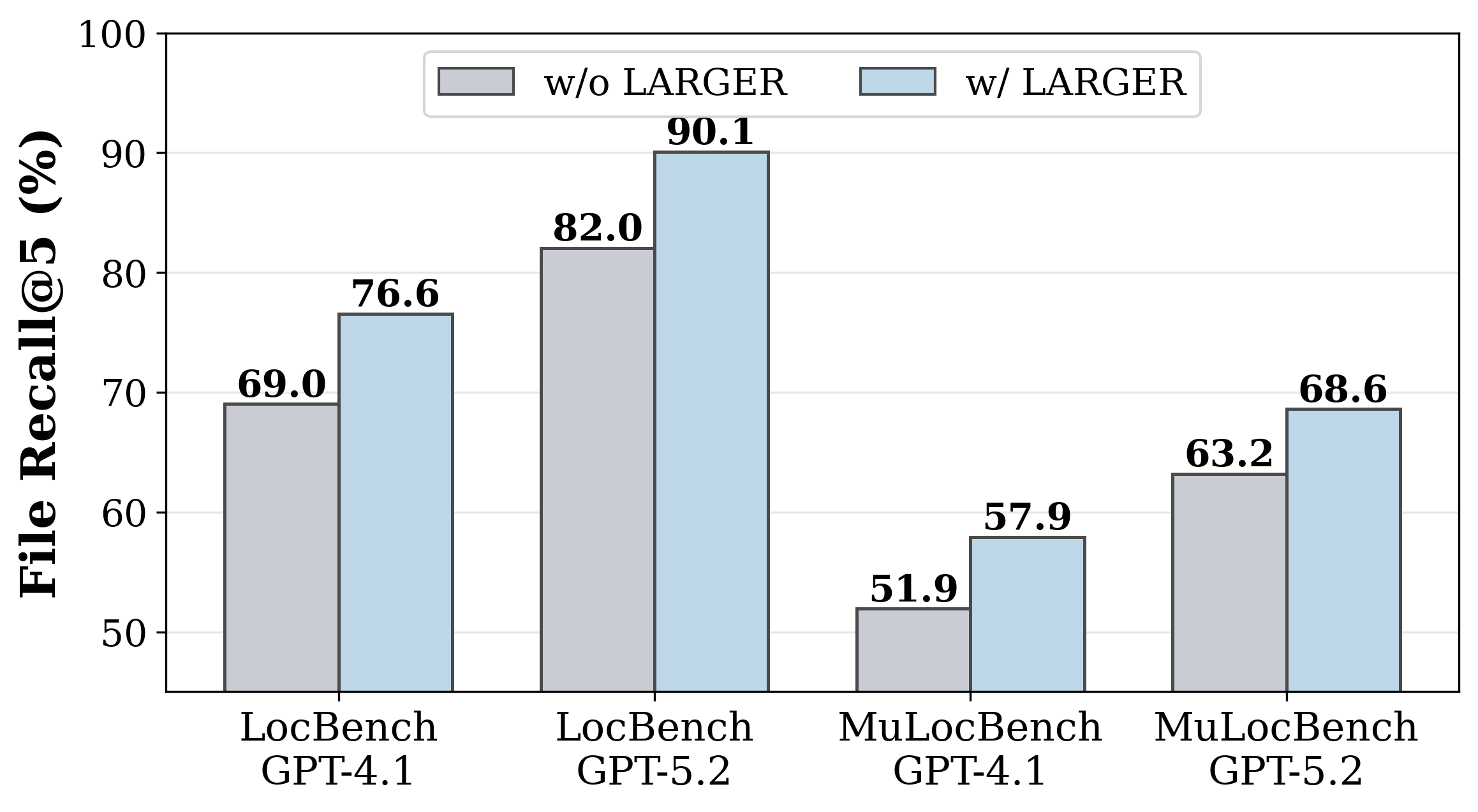}
    \caption{File-level Recall@5 with and without LARGER.}
    \label{fig:compare:a}
  \end{subfigure}
  \hfill
  \begin{subfigure}[t]{0.325\textwidth}
    \centering
    \includegraphics[width=\linewidth]{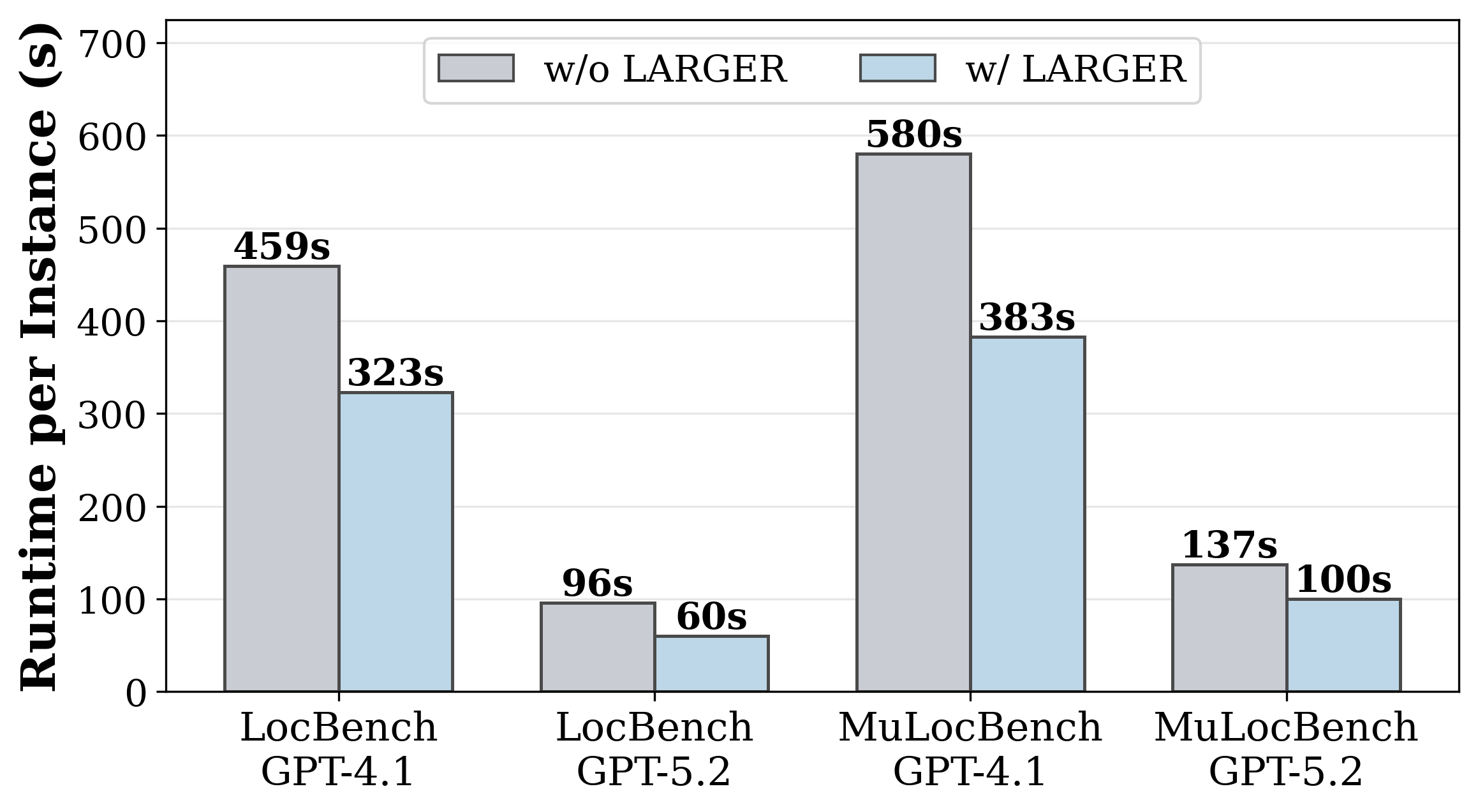}
    \caption{Mean wall-clock runtime per instance with and without LARGER.}
    \label{fig:compare:b}
  \end{subfigure}
  \begin{subfigure}[t]{0.335\textwidth}
    \centering
    \includegraphics[width=\linewidth]{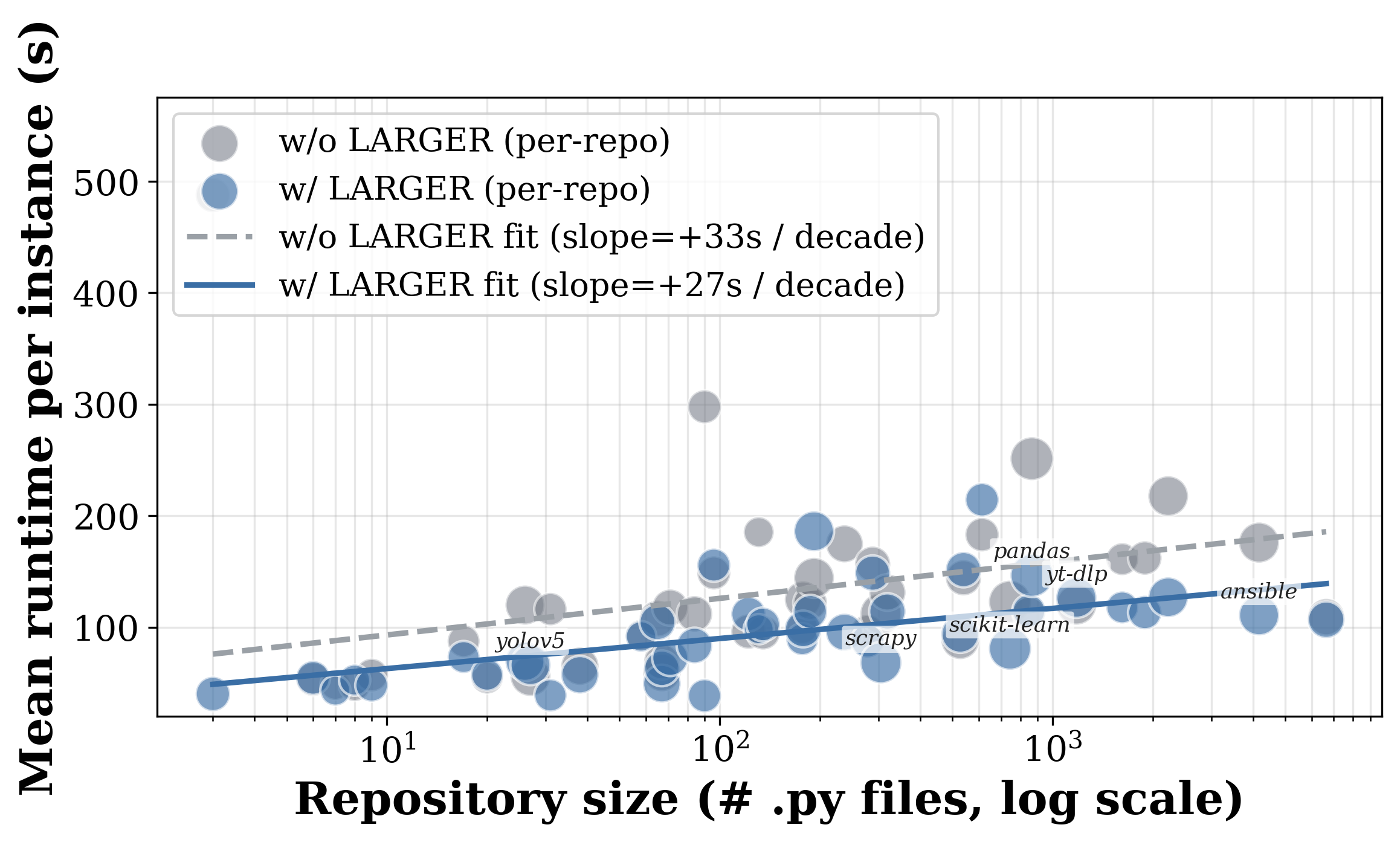}
    \caption{Repository-level runtime scaling versus repository size.}
    \label{fig:compare:c}
  \end{subfigure}
  \hfill
    \caption{Effect of enabling LARGER in a matched CLI agent. (a) File-level Recall@5 across LocBench and MuLocBench under GPT-4.1 and GPT-5.2. (b) Mean wall-clock time per instance for the same settings. (c) Repository-level runtime versus repository size; point size is proportional to issue count and lines are issue-weighted log-linear fits.}
  \label{fig:main:compare}
\end{figure*}

\paragraph{Toggling LARGER on/off lifts recall and lowers runtime.}
\hyperref[fig:main:compare]{Figure~\ref*{fig:main:compare}} complements the tables with matched ablations of LARGER itself. \hyperref[fig:compare:a]{Figures~\ref*{fig:compare:a} and~\ref*{fig:compare:b}} show that enabling LARGER improves Recall@5 \emph{and} reduces wall-clock time across both LocBench and MuLocBench under GPT-4.1 and GPT-5.2. \hyperref[fig:compare:c]{Figure~\ref*{fig:compare:c}} shows the same trend at repository scale: runtime rises with repository size for both settings, but the LARGER fit is shallower ($+27$\,s\,/\,decade vs.\ $+33$\,s\,/\,decade) and remains below the non-LARGER fit across the observed range.

\paragraph{Integration matters more than raw graph access.}
As shown in \hyperref[tab:main_results]{Table~\ref*{tab:main_results}}, CoSIL and LocAgent both trail the strongest non-graph baselines despite using structural information, whereas LARGER exceeds them. Graph evidence therefore helps only when delivered through a compact interface inside the agent's existing lexical loop rather than through fragmented or expensive traversal.

\paragraph{Localization gains transfer to downstream SWE-Atlas tasks.} Beyond localization, we evaluate LARGER on two downstream agentic software-engineering tasks from SWE-Atlas (Codebase QnA and Test Writing), which require multi-file reasoning, execution, and system-level understanding rather than isolated code edits. As \hyperref[tab:sweatlas]{Table~\ref*{tab:sweatlas}} reports, LARGER exceeds Claude Code by $+1.05$ on Codebase QnA and $+1.11$ on Test Writing, with larger margins against the same-backbone Codex baseline ($+2.42$ and $+5.56$, respectively); the retrieval gain therefore propagates beyond localization into multi-file reasoning and execution-level understanding.

\noindent
\begin{minipage}[t]{0.35\textwidth}
  \vspace{0pt}
  \centering
  \small
  \begin{tabular}{lcc}
  \toprule
  \textbf{Method} & QnA & TW\\
  \midrule
Codex & 29.83 & 32.22 \\
mini-swe-agent & 25.81 & 27.78 \\
Claude Code$^{*}$ & 31.20 &  36.67 \\
  \midrule
\textbf{LARGER$_{\text{Fixed}}$} & 32.25 & 37.78 \\
  \bottomrule
  \end{tabular}
  \captionof{table}{SWE-Atlas results (higher is better). QnA: Codebase QA. TW: Test Writing. All methods use GPT-5.2 except Claude Code$^{*}$, which uses Claude Opus 4.6.}
  \label{tab:sweatlas}
\end{minipage}
 \hfill
  \hfill
   \hfill
\begin{minipage}[t]{0.62\textwidth}
  \vspace{0pt}
  \centering
  \small
  \begin{tabular}{lcc}
  \toprule
  \textbf{Configuration} & \textbf{Acc@5} & \textbf{$\Delta$ \%} \\
  \midrule
  Full LARGER & 55.7 & -- \\
  \midrule
  w/o Graph expansion & 48.2 & $-13.5\%$ \\
  w/o Confidence scoring & 53.1 & $-4.7\%$ \\
  w/o Community detection & 53.4 & $-4.1\%$ \\
  \midrule
  w/o All modules & 46.8 & $-16.0\%$ \\
  \bottomrule
  \end{tabular}
  \captionof{table}{Component ablation on MuLocBench (Acc@5, \%) by removing retrieval components from the full LARGER system. $\Delta$ reports the relative percentage drop in Acc@5 with respect to the full system.}
  \label{tab:ablation}
\end{minipage}%

\paragraph{LARGER beats Codex on accuracy, runtime, and tokens.}
\label{sec:efficiency}
At the fixed operating point, LARGER improves on Codex along all three cost axes simultaneously: MuLocBench Acc@5 rises from 50.0 to 55.7 while mean wall-clock time drops from 139.9\,s to 99.9\,s and mean tokens from 521.8K to 353K; on LocBench, Acc@5 jumps from 74.1 to 87.0 with wall-clock time roughly halved (60.0\,s vs.\ 129.2\,s). The tuned setting trades a modest runtime increase for further accuracy (60.0 Acc@5 on MuLocBench, 89.1 on LocBench). The resulting accuracy--efficiency frontier is visualized in \hyperref[fig:accuracy_efficiency]{Figure~\ref*{fig:accuracy_efficiency}} (\hyperref[app:cost_decomp]{Appendix~\ref*{app:cost_decomp}}), where per-method tables, comparisons against cheaper baselines, and the LARGER cost decomposition are also reported.

\subsection{Ablation Study}
\label{sec:ablation}

To isolate the contribution of each component in the LARGER pipeline, we conduct an ablation study on the three retrieval components that most directly govern online graph exposure, together with one all-modules removal. All ablations are evaluated on MuLocBench with the same backbone LLM under the ablation-run configuration. \hyperref[tab:ablation]{Table~\ref*{tab:ablation}} reports Acc@5 and the absolute point drop relative to the corresponding full-system ablation run.

Removing graph expansion causes the largest single-component drop, from 55.7 to 48.2 Acc@5 ($-13.5\%$), confirming that explicit structural neighbor exposure is the primary source of LARGER's gain.
Once graph evidence is available, removing confidence scoring and community detection produces additional drops of $4.7\%$ and $4.1\%$, showing that noise control and subsystem-level priors materially improve retrieval quality on top of raw expansion.
Removing all graph modules produces the largest overall degradation (46.8 Acc@5, $-16.0\%$), and the all-modules drop exceeds any single-component drop, so the gains stack rather than substitute.

\section{Conclusion}
\label{sec:conclusion}
We recast repository context localization for CLI coding agents as \emph{Lexically Anchored Structural Localization}. LARGER realizes this view by pairing a multi-language, AST-based, confidence-scored graph with an active-set operator that anchors on the agent's own lexical queries and surfaces only confidence-filtered local neighborhoods, so structural evidence arrives inside the existing search loop without new tools or fragmented traversal. Across four benchmarks spanning localization, test writing, and codebase question answering, this design lifts MuLocBench file-level Acc@5 by up to 10 points over the strongest agent baseline while shifting the accuracy--efficiency frontier in wall-clock time and tokens, with ablations attributing the gain to the complementary effects of graph expansion, confidence scoring, and community priors.

\bibliographystyle{unsrtnat}
\bibliography{custom}

\appendix

\newpage
\section{Theoretical Analysis of LARGER}
\label{app:theory}

This appendix provides the proofs for \hyperref[thm:recall-dominance]{Theorem~\ref*{thm:recall-dominance}}, \hyperref[prop:step-dominance]{Proposition~\ref*{prop:step-dominance}}, and \hyperref[cor:token-cost]{Corollary~\ref*{cor:token-cost}} declared in \hyperref[sec:inner-loop]{Section~\ref*{sec:inner-loop}}, together with the supporting lemmas, assumption, and a finer quantitative discovery result. Throughout, $\mathcal{C}_t^{\mathrm{lex}}$ denotes the agent context produced by running the same policy $\pi$ and query stream with $\Gamma \equiv \emptyset$ (lexical-only retrieval), and $\mathcal{C}_t^{\mathrm{LARGER}}$ denotes the context produced by LARGER with $\Gamma = \Gamma$. We write $Y \subseteq V_c$ for the ground-truth target set and define $\mathrm{Recall}(\mathcal{C}, Y) = |\mathcal{C} \cap Y|/|Y|$. The compression operator $\Pi(\cdot)$ is assumed to be monotone in its input set, i.e., $\mathcal{C} \subseteq \mathcal{C}'$ implies $\Pi(\mathcal{C}) \subseteq \Pi(\mathcal{C}')$, and to preserve any subset that already fits within the bounded representation.

\subsection{Monotone Context Augmentation}

\begin{lemma}[Monotone context augmentation]
\label{lem:monotone}
If $\Pi$ preserves prior context (does not evict elements that already fit in the bounded representation), then $\mathcal{C}_{t-1} \subseteq \mathcal{C}_t$ for all $t \geq 1$.
\end{lemma}

\begin{proof}
By the update rule, $\mathcal{C}_t = \Pi\bigl(\mathcal{C}_{t-1} \cup M_t \cup \Gamma(M_t,q,\mathcal{C}_{t-1})\bigr)$. The argument of $\Pi$ is a superset of $\mathcal{C}_{t-1}$, and by the non-evicting assumption $\Pi$ retains every element of $\mathcal{C}_{t-1}$, yielding $\mathcal{C}_{t-1} \subseteq \mathcal{C}_t$.
\end{proof}

\subsection{Proof of Theorem~\ref{thm:recall-dominance} (Recall Dominance)}

\begin{proof}[Proof of Theorem~\ref{thm:recall-dominance}]
By induction on $t$. Base case $t=0$: both regimes start with $\mathcal{C}_0 = \emptyset$, and the inclusion is trivial. Inductive step: suppose $\mathcal{C}_{t-1}^{\mathrm{lex}} \subseteq \mathcal{C}_{t-1}^{\mathrm{LARGER}}$. With identical $\pi$ and $q_t$ across regimes, both regimes produce the same lexical matches $M_t$. The non-graph regime updates to $\Pi(\mathcal{C}_{t-1}^{\mathrm{lex}} \cup M_t)$, while LARGER updates to $\Pi(\mathcal{C}_{t-1}^{\mathrm{LARGER}} \cup M_t \cup \Gamma(M_t,q,\mathcal{C}_{t-1}^{\mathrm{LARGER}}))$, whose argument is a superset of $\mathcal{C}_{t-1}^{\mathrm{lex}} \cup M_t$. Monotonicity of $\Pi$ in its input set yields $\mathcal{C}_t^{\mathrm{lex}} \subseteq \mathcal{C}_t^{\mathrm{LARGER}}$. Since $\mathrm{Recall}(\cdot, Y)$ is non-decreasing in its first argument, the recall inequality follows.
\end{proof}

\subsection{Bounded Per-Step Augmentation}

\begin{lemma}[Bounded per-step augmentation]
\label{lem:bounded-aug}
Suppose the lexical matcher returns at most $m$ anchors per step, $\Gamma$ retains at most $k$ neighbors per anchor, and each rendered graph node consumes at most $L_{\mathrm{node}}^{\max}$ tokens. Then for every step $t$,
\[
\bigl|\Gamma(M_t, q, \mathcal{C}_{t-1})\bigr| \;\leq\; m\,k,
\]
and the per-step token overhead introduced by $\Gamma$ beyond lexical retrieval is at most $\Delta = m\,k\,L_{\mathrm{node}}^{\max}$. The cumulative augmentation after $T$ steps satisfies
\[
\Bigl| \bigcup_{t=1}^{T} \bigl(M_t \cup \Gamma(M_t, q, \mathcal{C}_{t-1})\bigr) \Bigr|
\;\leq\;
\sum_{t=1}^{T} |M_t| \;+\; T\,m\,k.
\]
The bounds are independent of repository size $|V_c|$ and graph maximum degree.
\end{lemma}

\begin{proof}
By construction, $\Gamma(M_t, q, \mathcal{C}_{t-1}) = \bigcup_{v \in M_t} \mathcal{N}_{K,t}^{*}(v)$, where each $\mathcal{N}_{K,t}^{*}(v)$ is a $\mathrm{Top\text{-}k}$ selection from $\mathcal{N}_K(v)$, hence $|\mathcal{N}_{K,t}^{*}(v)| \leq k$. With $|M_t| \leq m$, the per-step bound follows by counting and the cumulative bound by summing over $t$. The token bound follows from the $L_{\mathrm{node}}^{\max}$ assumption on rendered node size. None of these steps reference $|V_c|$ or graph degree.
\end{proof}

This bound is a property of the construction parameters and does \emph{not} translate into an asymptotic wall-clock guarantee, because per-step LLM interaction cost dominates the actual runtime. We use it only as input to the token-cost analysis below.

\subsection{Reachability Assumption}

\begin{assumption}[$(K,\theta)$-reachability]
\label{ass:reachability}
Every relevant node $y \in Y$ is reachable from some realized lexical anchor $v \in \bigcup_t M_t$ within $K$ hops along edges whose confidence weight $\omega(\cdot)$ exceeds the threshold $\theta$.
\end{assumption}

This assumption is mild in practice because (a) lexical anchors are abundant in code repositories due to identifier-rich queries, and (b) confidence-weighted $K$-hop neighborhoods with the values in \hyperref[app:confidence]{Appendix~\ref*{app:confidence}} cover the standard caller--callee, import, inheritance, test, and documentation links along which related code typically lies.

\subsection{Proof of Proposition~\ref{prop:step-dominance} (Step-Count Dominance)}

\begin{proof}[Proof of Proposition~\ref{prop:step-dominance}]
By \hyperref[thm:recall-dominance]{Theorem~\ref*{thm:recall-dominance}}, $\mathrm{Recall}(\mathcal{C}_t^{\mathrm{LARGER}}, Y) \geq \mathrm{Recall}(\mathcal{C}_t^{\mathrm{lex}}, Y)$ at every step $t$. Hence any threshold $R$ first reached by the non-graph regime at step $T^{\mathrm{lex}}(R)$ is reached by LARGER no later than $T^{\mathrm{lex}}(R)$.
\end{proof}

\paragraph{Remark on the fixed-stream assumption.}
Step-count dominance compares the two regimes on the same query stream $\{q_t\}$. In deployment, the agent generates $q_t$ adaptively from observations, so the LARGER regime may issue different $q_t$ from the lexical-only regime. Because LARGER's observations are supersets of the lexical ones, the additional information typically leads to more targeted subsequent queries; the step savings observed empirically in \hyperref[sec:efficiency]{Section~\ref*{sec:efficiency}} can therefore exceed the step-count dominance bound. The same-stream assumption thus yields a \emph{conservative} bound, not a best-case one.

\subsection{Proof of Corollary~\ref{cor:token-cost} (Conditional Token-Cost Dominance)}

\begin{proof}[Proof of Corollary~\ref{cor:token-cost}]
By \hyperref[lem:bounded-aug]{Lemma~\ref*{lem:bounded-aug}}, the per-step token overhead of $\Gamma$ beyond the non-graph regime is at most $\Delta = m\,k\,L_{\mathrm{node}}^{\max}$. The total token cost of LARGER to reach recall $R^*$ is therefore bounded by $T^{\mathrm{LARGER}}(R^*)\,(C_{\mathrm{step}}^{\mathrm{lex}} + \Delta)$, while the non-graph regime requires $T^{\mathrm{lex}}(R^*)\,C_{\mathrm{step}}^{\mathrm{lex}}$. Setting the former strictly less than the latter and rearranging yields the displayed condition.
\end{proof}

We do not claim unconditional wall-clock dominance: the corollary bounds only the token-budget contribution, while wall-clock time also depends on network latency, serving infrastructure, and prompt-cache behavior outside the scope of this analysis.

\subsection{Proof of Theorem~\ref{thm:quant-discovery} (Quantitative Discovery Advantage)}

\begin{proof}[Proof of Theorem~\ref{thm:quant-discovery}]
By \hyperref[thm:recall-dominance]{Theorem~\ref*{thm:recall-dominance}}, $\mathcal{C}_T^{\mathrm{lex}} \subseteq \mathcal{C}_T^{\mathrm{LARGER}}$, so $|\mathcal{C}_T^{\mathrm{LARGER}} \cap Y| - |\mathcal{C}_T^{\mathrm{lex}} \cap Y| = |H_T|$. Dividing by $|Y|$ yields the equality. The non-emptiness of $H_T$ under \hyperref[ass:reachability]{Assumption~\ref*{ass:reachability}} follows because any structurally reachable but lexically hidden $y$ enters $\mathcal{C}_T^{\mathrm{LARGER}}$ via $\Gamma$ at the step in which a corresponding anchor is realized, while remaining outside $\mathcal{C}_T^{\mathrm{lex}}$ by construction.
\end{proof}

\subsection{Discussion: Efficiency Comes from Fewer Interactions, Not Cheaper Steps}

The combination of \hyperref[thm:recall-dominance]{Theorem~\ref*{thm:recall-dominance}}, \hyperref[prop:step-dominance]{Proposition~\ref*{prop:step-dominance}}, \hyperref[lem:bounded-aug]{Lemma~\ref*{lem:bounded-aug}}, and \hyperref[cor:token-cost]{Corollary~\ref*{cor:token-cost}} establishes the following picture. LARGER is not necessarily cheaper than lexical-only retrieval on a per-step basis: each step may inject up to $\Delta$ extra tokens to render graph augmentation. The efficiency gain comes instead from \emph{reducing the number of interaction steps required to reach task-relevant evidence}, while the per-step augmentation overhead is bounded independently of the repository size. The fixed-query-stream comparison is conservative; in deployment, the richer observations LARGER provides typically lead the agent to issue different and more targeted subsequent queries, so the empirical step savings can exceed the bound proved here. This is consistent with the wall-clock and token reductions reported in \hyperref[sec:efficiency]{Section~\ref*{sec:efficiency}}.

\section{Full Experimental Results}
\label{app:full_results}

\begin{table*}[t]
\centering
\small
\resizebox{\textwidth}{!}{
\begin{tabular}{l cccccc cccccc}
\toprule
\multirow{2}{*}{\textbf{Method}} & \multicolumn{6}{c}{\textbf{File-Level}} & \multicolumn{6}{c}{\textbf{Function-Level}} \\
\cmidrule(lr){2-7} \cmidrule(lr){8-13}
 & Hit@1 & Acc@1 & Acc@5 & Rec@5 & MAP & MRR & Hit@1 & Acc@1 & Acc@5 & Rec@5 & MAP & MRR \\
\midrule
BM25 & 33.9 & 23.4 & 49.3 & 55.7 & 39.7 & 47.0 & 24.5 & 13.6 & 24.1 & 29.9 & 24.0 & 31.6 \\
Agentless & 72.1 & 56.1 & 68.9 & 76.6 & 68.9 & 78.8 & 41.3 & 24.1 & 39.8 & 48.6 & 41.1 & 49.5 \\
SWE-agent & 78.2 & 60.0 & 70.9 & 79.3 & 74.4 & 83.5 & 53.9 & 30.5 & 44.5 & 54.9 & 48.9 & 60.4 \\
OpenHands & 41.3 & 31.1 & 35.9 & 40.1 & 38.5 & 42.8 & 35.5 & 21.6 & 29.3 & 35.1 & 32.3 & 38.0 \\
Codex & \textbf{82.3} & \textbf{63.2} & \textbf{74.1} & \textbf{82.8} & \textbf{78.1} & \textbf{86.4} & \textbf{70.5} & \textbf{42.3} & \textbf{57.9} & \textbf{70.3} & \textbf{64.7} & \textbf{76.9} \\
mini-swe-agent & 80.6 & 60.6 & 70.7 & 79.7 & 75.2 & 84.8 & 59.1 & 34.6 & 52.2 & 63.4 & 55.2 & 66.5 \\
CoSIL & 70.7 & 54.6 & 68.2 & 75.7 & 67.6 & 77.2 & 62.0 & 35.2 & 49.1 & 60.8 & 53.9 & 67.7 \\
LocAgent & 74.5 & 56.0 & 65.3 & 73.8 & 68.7 & 78.9 & 9.7 & 5.0 & 11.7 & 15.8 & 13.0 & 16.8 \\
\bottomrule
\end{tabular}
}
\caption{Full baseline results on LocBench~\citep{chen2025locagent} (\%). File-level and function-level localization metrics for the baseline methods, complementing \hyperref[tab:main_results]{Table~\ref*{tab:main_results}} with additional metrics (Hit@$K$, MAP, MRR). All baselines use GPT-5.2. Hit@$K$: any ground-truth entity in top-$K$. Acc@$K$: all ground-truth entities in top-$K$. Rec@$K$: recall at $K$.}
\label{tab:locbench_full}
\end{table*}

\begin{table*}[t]
\centering
\small
\resizebox{\textwidth}{!}{
\begin{tabular}{l cccccc cccccc}
\toprule
\multirow{2}{*}{\textbf{Method}} & \multicolumn{6}{c}{\textbf{File-Level}} & \multicolumn{6}{c}{\textbf{Function-Level}} \\
\cmidrule(lr){2-7} \cmidrule(lr){8-13}
 & Hit@1 & Acc@1 & Acc@5 & Rec@5 & MAP & MRR & Hit@1 & Acc@1 & Acc@5 & Rec@5 & MAP & MRR \\
\midrule
BM25 & 28.5 & 10.6 & 24.5 & 35.3 & 27.2 & 38.9 & 11.0 & 3.9 & 10.7 & 15.0 & 10.8 & 16.2 \\
Agentless & 49.6 & 20.4 & 28.9 & 42.0 & 36.4 & 55.3 & 17.4 & 7.4 & 14.4 & 19.6 & 15.2 & 22.3 \\
SWE-agent & 50.7 & 21.0 & 31.3 & 44.5 & 40.0 & 56.5 & 21.9 & 10.4 & 15.7 & 21.2 & 18.4 & 25.5 \\
OpenHands & 35.1 & 16.3 & 26.4 & 33.7 & 30.6 & 38.0 & 16.3 & 7.5 & 12.4 & 16.3 & 14.1 & 18.8 \\
Codex & \textbf{64.8} & \textbf{28.0} & \textbf{50.0} & \textbf{65.1} & \textbf{57.3} & \textbf{72.6} & 26.5 & 12.5 & \textbf{20.3} & \textbf{27.2} & \textbf{23.4} & \textbf{31.8} \\
mini-swe-agent & 61.7 & 26.9 & 46.6 & 61.4 & 53.2 & 69.2 & \textbf{27.5} & \textbf{12.6} & 20.2 & 26.7 & 22.5 & 31.4 \\
CoSIL & 47.0 & 18.1 & 27.0 & 40.1 & 34.1 & 52.9 & 21.3 & 8.4 & 13.0 & 18.2 & 15.5 & 24.0 \\
LocAgent & 47.8 & 17.9 & 25.7 & 39.0 & 34.7 & 53.1 & 4.5 & 1.7 & 4.6 & 6.4 & 5.2 & 7.6 \\
\bottomrule
\end{tabular}
}
\caption{Full baseline results on MuLocBench~\citep{zhang2025mulocbench} (\%). File-level and function-level localization metrics for the baseline methods, complementing \hyperref[tab:main_results]{Table~\ref*{tab:main_results}} with additional metrics (Hit@$K$, MAP, MRR). All baselines use GPT-5.2. MuLocBench targets multi-file issues from 46 repositories.}
\label{tab:mulocbench_full}
\end{table*}

\subsection{Efficiency Tables and Cost Decomposition}
\label{app:cost_decomp}

This appendix collects the full per-method efficiency comparison summarized in \hyperref[sec:efficiency]{Section~\ref*{sec:efficiency}} (\hyperref[tab:runtime]{Tables~\ref*{tab:runtime} and~\ref*{tab:tokens}}) and a detailed decomposition of LARGER's per-instance token usage and cost. All numbers are computed from agent execution traces captured during the headline runs (\textsc{enhanced\_nexus\_v2\_extended\_full} for MuLocBench, \textsc{locbench\_v2ext\_aligned} for LocBench), using the per-step \texttt{step\_finish} telemetry emitted by the agent runtime. Coverage is 1097/1098 instances on MuLocBench and 374/560 on LocBench (the remaining LocBench instances failed with API authentication errors and are excluded).

\paragraph{Comparison against cheaper baselines and the alternative graph design.}
Agentless, CoSIL, and mini-swe-agent remain cheaper in absolute terms (\hyperref[tab:runtime]{Tables~\ref*{tab:runtime}--\ref*{tab:tokens}}) but are also markedly less accurate, occupying the lower-left of \hyperref[fig:accuracy_efficiency]{Figure~\ref*{fig:accuracy_efficiency}}. The other graph-enhanced agent, LocAgent, is dominated on both axes: unbounded global traversal is both more expensive and less effective than LARGER's tightly integrated graph augmentation. Together, the fixed and tuned LARGER settings occupy the upper frontier in both panels of \hyperref[fig:accuracy_efficiency]{Figure~\ref*{fig:accuracy_efficiency}}.

\begin{figure*}[t]
\centering
\includegraphics[width=\textwidth]{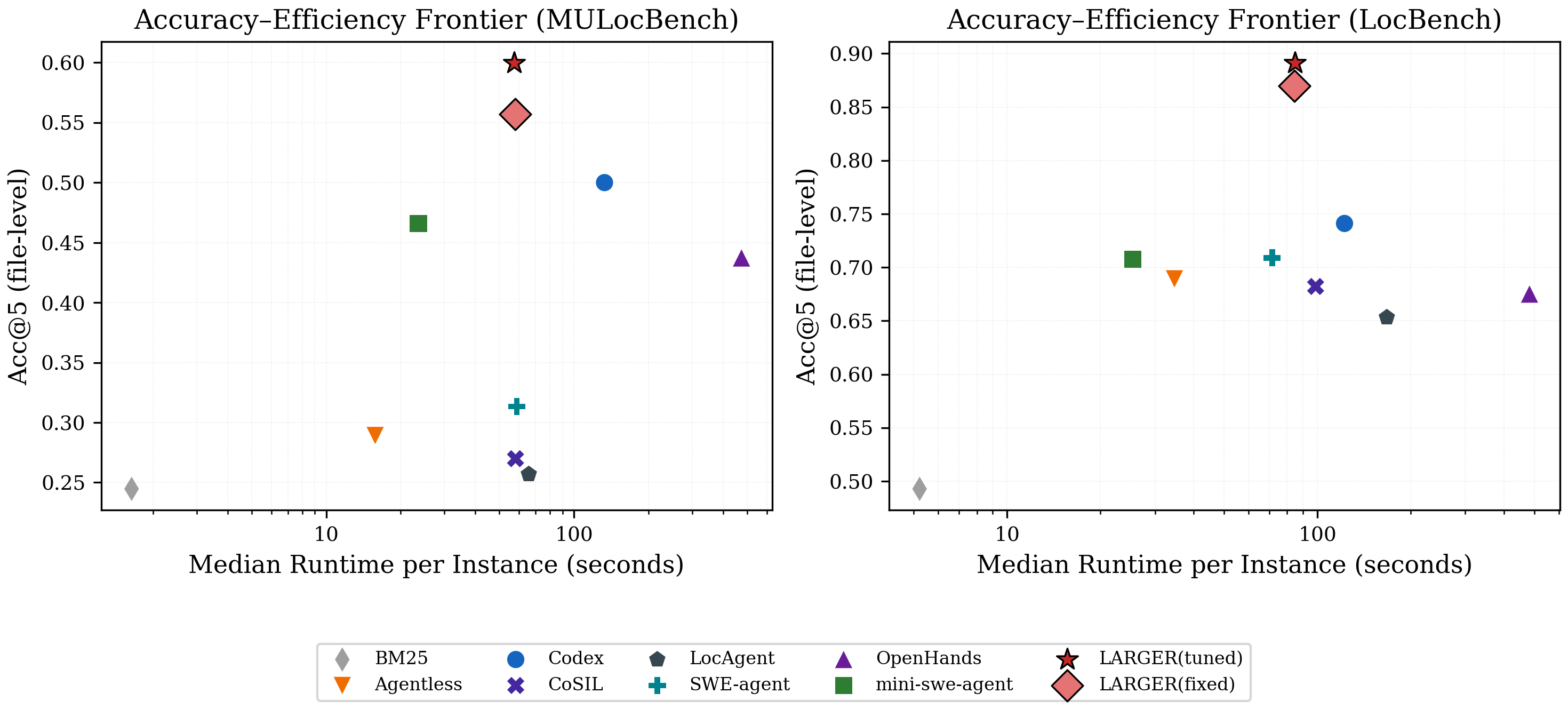}
\caption{\textbf{Accuracy--efficiency frontier on MuLocBench and LocBench.} Each point is one method, with file-level Acc@5 on the y-axis and median runtime per instance on the x-axis. The two LARGER operating points lie on the upper frontier in both benchmarks.}
\label{fig:accuracy_efficiency}
\end{figure*}

\begin{table}[t]
\centering
\small
\begin{tabular}{l rr rr}
\toprule
\multirow{2}{*}{\textbf{Method}} & \multicolumn{2}{c}{\textbf{MuLocBench}} & \multicolumn{2}{c}{\textbf{LocBench}} \\
\cmidrule(lr){2-3} \cmidrule(lr){4-5}
 & Mean (s) & Median (s) & Mean (s) & Median (s) \\
\midrule
BM25 & 6.4 & 1.6 & 8.0 & 5.2 \\
Agentless & 39.3 & 15.8 & 48.5 & 34.7 \\
SWE-agent & 72.3 & 58.8 & 79.9 & 71.5 \\
OpenHands & 635.7 & 474.6 & 710.7 & 483.1 \\
Codex & 139.9 & 132.8 & 129.2 & 122.4 \\
Claude Code & 98.3 & 81.7 & 69.4 & 54.4 \\
mini-swe-agent & 32.3 & 23.6 & 33.5 & 25.4 \\
CoSIL & 109.7 & 58.1 & 151.9 & 98.8 \\
LocAgent & 263.4 & 65.7 & 623.6 & 167.6 \\
\midrule
\textbf{LARGER (Ours)} & 99.9 & 71.6 & 60.0 & 40.9 \\
\bottomrule
\end{tabular}
\caption{Wall-clock runtime per instance (seconds). All in-house methods run on the same hardware with GPT-5.2; proprietary baselines are reported under their native interfaces.}
\label{tab:runtime}
\end{table}

\begin{table}[t]
\centering
\small
\resizebox{\textwidth}{!}{
\begin{tabular}{l ll ll ll l}
\toprule
\multirow{2}{*}{\textbf{Method}} & \multicolumn{3}{c}{\textbf{MuLocBench}} & \multicolumn{3}{c}{\textbf{LocBench}} \\
\cmidrule(lr){2-4} \cmidrule(lr){5-7}
 & Tokens/inst & API calls & Cost/inst & Tokens/inst & API calls & Cost/inst \\
\midrule
Agentless & 17.8K & -- & \$0.03$^\ddagger$ & 26.3K & -- & \$0.05$^\ddagger$ \\
Codex & 521.8K & -- & \$0.98$^\ddagger$ & 591.4K & -- & \$1.10$^\ddagger$ \\
CoSIL & 18.7K & -- & \$0.04$^\ddagger$ & 25.0K & -- & \$0.05$^\ddagger$ \\
LocAgent & 711.5K & -- & \$1.26$^\ddagger$ & 940.7K & -- & \$1.67$^\ddagger$ \\
SWE-agent & 96.1K & 14.9 & \$0.12 & 105.0K & 15.3 & \$0.14 \\
OpenHands & -- & 22.9$^\dagger$ & -- & -- & 23.6$^\dagger$ & -- \\
mini-swe-agent & -- & 15.6 & \$0.09 & -- & 16.3 & \$0.10 \\
\midrule
\textbf{LARGER (Ours)} & 353K & 13.0 & \$0.44$^\ddagger$ & 452K & 14.5 & \$0.56$^\ddagger$ \\
\bottomrule
\end{tabular}
}
\caption{Token usage and cost per instance. Tokens/inst: mean total tokens (prompt + completion). BM25 is excluded (no LLM calls). The LARGER row corresponds to the fixed configuration used in \hyperref[tab:main_results]{Table~\ref*{tab:main_results}}; its token telemetry covers 1097/1100 MuLocBench instances and 374/560 LocBench instances. ``--'' indicates the framework does not expose the metric. $^\dagger$OpenHands reports agent-loop iterations, not raw API calls. $^\ddagger$Cost is a coarse list-price estimate from the available framework token telemetry and should be interpreted only as an order-of-magnitude comparison. SWE-agent and mini-swe-agent costs are recorded directly via litellm.}
\label{tab:tokens}
\end{table}

\paragraph{Token decomposition.} \hyperref[tab:cost_decomp_tokens]{Table~\ref*{tab:cost_decomp_tokens}} reports the distribution of tokens consumed per instance, broken down by category. Both \emph{Input (cache miss)} and \emph{Input (cache hit)} are prompt tokens; the difference is whether they were served from the GPT-5.2 prompt cache: cache misses are billed at the full input rate, cache hits at $\sim$10\% of that rate. \emph{Output} are generated tokens, and \emph{Reasoning} are thinking-mode tokens billed at the output rate. The total prompt size at any call is \emph{Input (cache miss)} $+$ \emph{Input (cache hit)}. We report mean, median, and the inter-quartile range (P25--P75) to characterize the distribution, since per-issue cost varies considerably with repository size and issue complexity.

\begin{table}[h]
\centering
\small
\begin{tabular}{l rrrr | rrrr}
\toprule
\multirow{2}{*}{\textbf{Token category}} & \multicolumn{4}{c|}{\textbf{MuLocBench}} & \multicolumn{4}{c}{\textbf{LocBench}} \\
\cmidrule(lr){2-5} \cmidrule(lr){6-9}
 & Mean & Median & P25 & P75 & Mean & Median & P25 & P75 \\
\midrule
Input (cache miss) &  62.8K &  46.9K &  34.7K &  66.0K & 110.6K &  95.6K &  72.2K & 132.5K \\
Input (cache hit)  & 286.1K & 190.3K &  97.7K & 367.6K & 337.2K & 231.8K & 119.1K & 431.2K \\
Output             &   4.4K &   3.2K &   1.9K &   5.2K &   4.7K &   3.5K &   2.1K &   5.9K \\
Reasoning          &   3.1K &   2.0K &   0.9K &   3.6K &   3.3K &   2.1K &   1.0K &   4.3K \\
\midrule
\textbf{Total tokens} & \textbf{353.3K} & \textbf{255.0K} & \textbf{145.7K} & \textbf{446.5K} & \textbf{452.5K} & \textbf{339.4K} & \textbf{209.3K} & \textbf{543.4K} \\
\textbf{API calls (steps)} & \textbf{13.0} & \textbf{11.0} & \textbf{7.0} & \textbf{16.0} & \textbf{14.5} & \textbf{12.0} & \textbf{8.0} & \textbf{18.0} \\
\bottomrule
\end{tabular}
\caption{Per-instance token decomposition for LARGER. \emph{Input (cache miss)} and \emph{Input (cache hit)} together form the prompt tokens at each call, distinguished by whether they were served from the GPT-5.2 prompt cache; \emph{Reasoning} tokens are thinking-mode tokens billed at the output rate. The long upper tail of \emph{Input (cache hit)} reflects multi-step trajectories on large repositories where the system prompt and accumulated context are reused across calls. API calls counts \texttt{step\_finish} events emitted by the agent runtime per instance.}
\label{tab:cost_decomp_tokens}
\end{table}

\paragraph{Cost decomposition.} \hyperref[tab:cost_decomp_cost]{Table~\ref*{tab:cost_decomp_cost}} translates the token decomposition into US dollars using GPT-5.2 list pricing: \$1.75 per 1M cache-miss input tokens, \$0.175 per 1M cache-hit input tokens (the cached-input rate), and \$14.00 per 1M output and reasoning tokens. The bottom row of the table reports the cost actually billed by the runtime (which uses the same list rates); the small residual against the list-price reconstruction reflects sub-token rounding and minor accounting differences in the runtime's cost reporter. Despite consuming roughly half a million tokens per instance, the dominant cost contributor on MuLocBench is \emph{cache-miss input} (41\% of total) rather than cache hits, because the prompt cache is partially invalidated whenever the working set of files changes between agent steps. On LocBench the share shifts further toward cache-miss input (53\%), consistent with longer search trajectories on its larger repositories.

\begin{table}[h]
\centering
\small
\begin{tabular}{l rr | rr}
\toprule
\multirow{2}{*}{\textbf{Cost component (per instance)}} & \multicolumn{2}{c|}{\textbf{MuLocBench}} & \multicolumn{2}{c}{\textbf{LocBench}} \\
\cmidrule(lr){2-3} \cmidrule(lr){4-5}
 & USD & Share & USD & Share \\
\midrule
Input (cache miss) @ \$1.75/M  & \$0.110 & 41.4\% & \$0.193 & 53.1\% \\
Input (cache hit)  @ \$0.175/M & \$0.050 & 18.9\% & \$0.059 & 16.2\% \\
Output             @ \$14.00/M & \$0.062 & 23.2\% & \$0.066 & 18.0\% \\
Reasoning          @ \$14.00/M & \$0.044 & 16.5\% & \$0.046 & 12.7\% \\
\midrule
\textbf{Total (list-price reconstruction)} & \textbf{\$0.265} & 100\% & \textbf{\$0.364} & 100\% \\
Total (runtime-reported)       & \$0.251 & --    & \$0.347 & --    \\
\bottomrule
\end{tabular}
\caption{Per-instance cost decomposition for LARGER using GPT-5.2 list pricing. Shares are relative to the list-price reconstruction. Cache-miss input dominates because each agent step expands the visible context with new graph evidence, partially invalidating the prompt cache between calls.}
\label{tab:cost_decomp_cost}
\end{table}

\section{Dataset Details}
\label{app:datasets}

\begin{table}[h]
\centering
\small
\resizebox{\columnwidth}{!}{
\begin{tabular}{lccccl}
\toprule
\textbf{Benchmark} & \textbf{Instances} & \textbf{Repos} & \textbf{Languages} & \textbf{Task Type} & \textbf{Evaluation} \\
\midrule
LocBench & 560 & 5 & Python & Localization & Acc@$k$, Recall@$k$, NDCG@$k$ \\
MuLocBench & 1,100 & 46 & Python & Multi-file localization & Acc@$k$, Recall@$k$, NDCG@$k$ \\
SWE-Atlas TW & 90 & 11 & Py / Go / TS / C & Test writing & Mutation + rubric judge \\
SWE-Atlas QA & 124 & 11 & Py / Go / TS / C & Codebase QA & Rubric judge \\
\bottomrule
\end{tabular}
}
\caption{Summary of evaluation benchmarks used in our experiments.}
\label{tab:dataset_summary}
\end{table}

\paragraph{LocBench.}
LocBench (Loc-Bench\_V1)~\citep{chen2025locagent} is a code localization benchmark derived from real-world GitHub issues across five Python repositories. Each instance consists of a natural-language problem statement and a target repository at a specific commit. The ground truth specifies the set of files, modules (classes), and functions that require modification. We evaluate at three granularities: file-level (Acc@\{1,3,5\}), module-level (Acc@\{5,10\}), and function-level (Acc@\{5,10\}), along with Recall, NDCG, Precision, and MAP at the same cutoffs. LocBench isolates the localization subtask from patch generation, making it a direct test of retrieval quality without confounding effects from code synthesis.

\paragraph{MuLocBench.}
MuLocBench~\citep{zhang2025mulocbench} is a multi-file localization benchmark containing 1,100 issues from 46 popular Python repositories. Unlike single-file benchmarks, MuLocBench specifically targets issues whose resolution requires changes across multiple files, making structural retrieval signals (imports, call chains, co-change patterns) especially important for achieving high recall. We use the same localization metrics as LocBench. The 46 repositories span web frameworks (flask, django, fastapi), ML libraries (scikit-learn, transformers, pytorch), data tools (pandas, numpy), and utilities (requests, click, rich), providing diversity in repository size (2K--500K LOC), architectural style, and application domain.

\paragraph{SWE-Atlas Test Writing.}
SWE-Atlas Test Writing (SWE-Atlas TW)~\citep{scale2026sweatlas} is a benchmark of 90 test-generation tasks across 11 real-world repositories in four languages (Python, Go, TypeScript, C). Each task provides a repository at a specific commit, an instruction describing what tests to write, and a rubric of must-have and nice-to-have requirements. Evaluation is three-phase: (1)~Docker-based execution of the generated test suite with mutation testing to verify behavioral correctness, (2)~rubric-based LLM judging of test quality and coverage, and (3)~manifest correctness verification. The repositories include paperless-ngx, scapy, minio, grafana, kitty, k6, and others. This benchmark tests whether graph-augmented retrieval helps agents discover relevant implementation files, test utilities, and fixtures needed to write comprehensive tests in unfamiliar, multi-language codebases.

\paragraph{SWE-Atlas Codebase QA.}
SWE-Atlas Codebase QA (SWE-Atlas QA)~\citep{scale2026sweatlas} is a code understanding benchmark of 124 questions across the same 11 multi-language repositories. Each question requires the agent to explore the repository and provide an evidence-based answer with specific code references (file paths, line numbers, variable names). Questions are categorized into architecture and system design (44 instances), code onboarding (28), root-cause analysis (37), security (11), and API/integration usage (4). Evaluation uses rubric-based LLM judging, scoring each answer against must-have criteria. This benchmark directly tests whether structural graph context helps agents navigate and understand codebases they have not seen before, complementing the localization-focused benchmarks with an understanding-focused evaluation.

\section{Additional Empirical Analyses}
\label{app:empirical}

This appendix collects two empirical studies that complement the main results: a runtime comparison between commit-aware alignment and full graph reconstruction (\hyperref[app:alignment]{\S\ref*{app:alignment}}), and a per-repository hyperparameter sweep characterizing how the oracle-optimal $(k, \theta)$ varies with repository size (\hyperref[app:hyperparam]{\S\ref*{app:hyperparam}}).

\subsection{Commit-Aware Alignment Cost}
\label{app:alignment}

We expand on the cost analysis sketched in \hyperref[sec:outer-loop]{Section~\ref*{sec:outer-loop}}. Constructing $G_{c_0}$ from scratch parses every file in the repository at the reference commit $c_0$ and runs the heterogeneous extraction, confidence weighting, and community partitioning passes, incurring cost $O(|V_{c_0}| + |E_{c_0}|)$ in the resulting graph size plus parser overhead per file.

Aligning $G_{c_0}$ to a target commit $c$ touches only the per-commit diff. Let $\Delta(c) = V(c) \,\triangle\, V_{c_0}$ denote the symmetric difference of file sets at $c$ and $c_0$. Alignment performs three operations: (i)~filtering $V_{c_0} \cap V(c)$ from the persistent node table; (ii)~parsing $V_\Delta(c)$ to produce the new node and edge sets $V_\Delta, E_\Delta$; and (iii)~indexing the resulting subgraph into the sidecar storage (\hyperref[sec:implementation-details]{Appendix~\ref*{sec:implementation-details}}). Each operation runs in time linear in $|\Delta(c)|$ plus the local parse output, giving alignment cost
\[
T_{\mathrm{align}}(c \mid c_0) \;=\; O\bigl(|\Delta(c)|\bigr),
\]
which is independent of the unchanged portion of the graph at $c_0$.

Communities and other global structures are recomputed lazily, only when the cumulative diff exceeds a configurable threshold or when retrieval quality begins to degrade; between recomputations they are reused as soft priors. Since typical commits modify only a small fraction of files, alignment is orders of magnitude cheaper than full reconstruction. This makes branch switching and cross-commit evaluation tractable on large repositories where full reconstruction would dominate the inference budget.

\paragraph{Empirical validation.}
To verify that the asymptotic cost gap between full reconstruction and commit-aware alignment translates into a practically large speedup, we instrumented the construction pipeline against three large Python repositories drawn from MuLocBench: \texttt{pandas}, \texttt{transformers}, and \texttt{scikit-learn}. The reference graph $G_{c_0}$ is built once at each repository's recent \texttt{HEAD} using the \emph{v1\_full} configuration (test files dropped, MAX caps at 10, base edges only) and cached on disk. We then iterate over 15 historical \texttt{base\_commit}s per repository; each commit is the buggy state pinned by a MuLocBench instance, so the commits are real bug-fix targets rather than synthetic check-points and span several years of repository history. For every commit $c$, we measure: (i)~$T_{\mathrm{rebuild}}(c)$, the wall time of \texttt{build\_graph} followed by \texttt{build\_enhanced\_index} on the worktree checked out at $c$ (this is the LocAgent-style static path); (ii)~$T_{\mathrm{align}}(c\,|\,c_0)$, the wall time of \texttt{align\_index\_to\_commit} applied to the cached $G_{c_0}$ given the same worktree; and (iii)~$|\Delta(c)|$, the size of the file set the alignment algorithm actually processes, decomposed as $|\Delta(c)| = |V(c) \setminus V_{c_0}| + |V_{c_0} \setminus V(c)|$ (\emph{missing} plus \emph{stale}). Both pipelines use identical edge-set and language-coverage flags, so the comparison isolates the cost of restarting parsing, edge resolution, and community detection from scratch.

\begin{figure}[t]
\centering
\includegraphics[width=\linewidth]{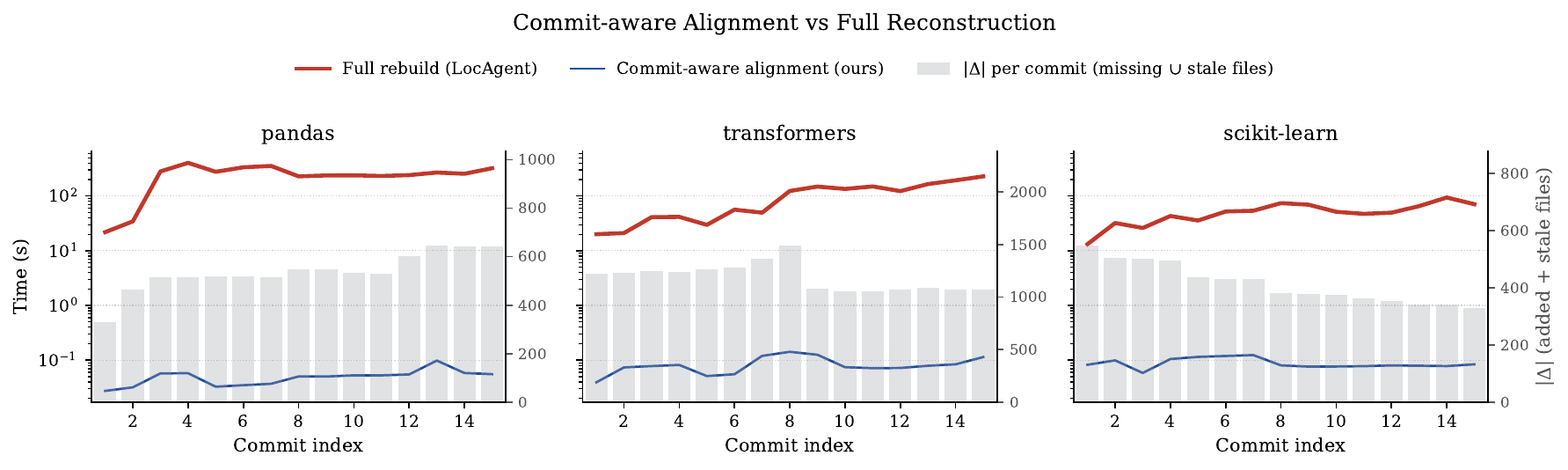}
\caption{Per-commit cost of full graph reconstruction vs.\ commit-aware alignment on MuLocBench base commits, log-scaled time axis. The red curve is the LocAgent-style static rebuild ($T_{\mathrm{rebuild}}$); the blue curve is alignment from a cached reference graph ($T_{\mathrm{align}}$); the shaded gray bars on the right axis report $|\Delta(c)|$, the file set the alignment algorithm processes at each commit (missing $\cup$ stale w.r.t.\ the cached $G_{c_0}$). Even when $|\Delta(c)|$ runs into the hundreds of files (commits drawn from years of history against a single cached reference), alignment stays in the tens-of-milliseconds range, while reconstruction grows with the absolute size of the repository at $c$. Time axis is log-scaled.}
\label{fig:alignment-efficiency}
\end{figure}

\hyperref[fig:alignment-efficiency]{Figure~\ref*{fig:alignment-efficiency}} shows the resulting trace. Aggregated across the 15 commits per repository, full reconstruction takes a median of 254\,s on \texttt{pandas} (max 402\,s), 123\,s on \texttt{transformers} (max 228\,s), and 51\,s on \texttt{scikit-learn} (max 94\,s), whereas alignment completes in median 53\,ms, 78\,ms, and 81\,ms respectively, with worst-case latencies under 150\,ms in all three repositories. The resulting per-commit speedup is $792\!-\!9528\times$ for \texttt{pandas} (median $4508\times$), $284\!-\!2292\times$ for \texttt{transformers} (median $1013\times$), and $159\!-\!1198\times$ for \texttt{scikit-learn} (median $605\times$). The asymmetry between repositories is driven primarily by absolute repository size at $c_0$: \texttt{pandas} is the largest of the three at recent \texttt{HEAD}, so its full-reconstruction baseline pays the most for the unchanged portion of the graph that alignment skips entirely.

The trace also confirms the predicted scaling regimes. Reconstruction time tracks the \emph{absolute} number of nodes and edges in $G_c$, which grows with each repository's history: the leftmost commits in the \texttt{pandas} and \texttt{transformers} panels execute in 20--35\,s because those very old commits ship a much smaller codebase, and the curves climb as the repository grows toward modern size. Alignment time is essentially flat across the same window, dependent on $|\Delta(c)|$ rather than $|V_c|$. The gray $|\Delta(c)|$ bars are large by design: we deliberately stress alignment by using \texttt{HEAD} as the single cached reference for commits drawn from across years of history, so the median $|\Delta|$ is 531 files for \texttt{pandas}, 1226 for \texttt{transformers} (dominated by 978 stale paths in $V_{c_0}\setminus V(c)$ from files that did not yet exist historically), and 383 for \texttt{scikit-learn}. Even under this regime, the slope of $T_{\mathrm{align}}$ as a function of $|\Delta(c)|$ is shallow because the dominant operation is regex-based import extraction over the missing-Python subset (\hyperref[app:alignment]{\S\ref*{app:alignment}}), and stale entries are scrubbed in time linear in the surviving entry count rather than the full graph.

The practical implication for our experimental pipeline is that cross-commit evaluation on MuLocBench (1,100 instances across 46 repositories, each pinned to a distinct \texttt{base\_commit}) is amortized by reusing one cached $G_{c_0}$ per repository rather than rebuilding 1,100 graphs. Naively rebuilding at every instance would, at the rates measured here, add a four-figure multiplier to the wall-clock cost of an evaluation sweep, enough to make repeated ablation runs prohibitive on the larger repositories. Alignment moves this overhead from the per-instance critical path to a one-time per-repository cost, leaving the instance-level budget dominated by agent inference rather than indexing.

\subsection{Hyperparameter Sensitivity}
\label{app:hyperparam}

To understand how sensitive LARGER's two retrieval-side hyperparameters (top-$k$ and the edge-confidence threshold $\theta$) are to repository scale, we sweep $(k,\theta)\in\{3,5,10,20\}\times\{0,0.5\}$ on every repository in MuLocBench and LocBench, and select the per-repository oracle pick maximizing Acc@5. \hyperref[fig:oracle_vs_size]{Figure~\ref*{fig:oracle_vs_size}} plots the oracle pick against repository size (log-scaled), with each point one repository and marker size proportional to its issue count.

\begin{figure}[t]
\centering
\includegraphics[width=\linewidth]{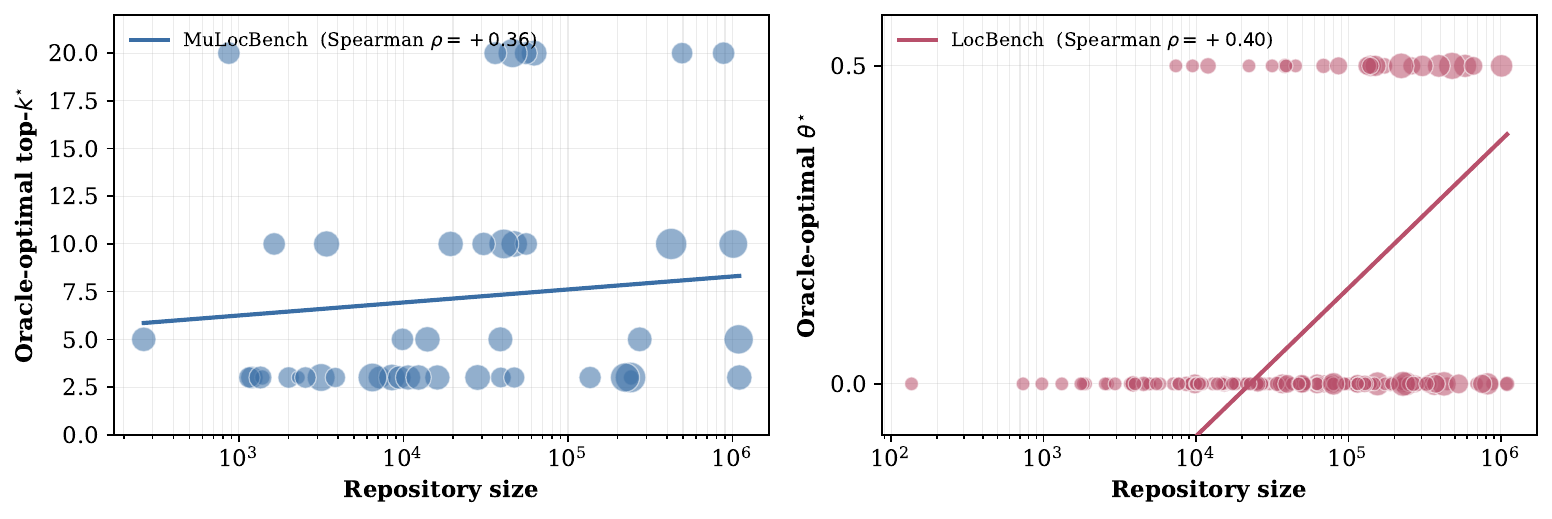}
\caption{\textbf{Oracle hyperparameters versus repository size.} Each point is one repository; marker size is proportional to its issue count. Solid lines are issue-weighted least-squares fits in $\log_{10}(\text{LOC})$. \textbf{(a)} On MuLocBench (multi-file edits), the oracle top-$k^{\star}$ trends upward with repository size (Spearman $\rho{=}{+}0.36$): larger repos benefit from wider neighborhoods. \textbf{(b)} On LocBench, the oracle confidence threshold $\theta^{\star}$ trends upward with repository size ($\rho{=}{+}0.40$): larger repos benefit from stricter filtering of low-confidence edges.}
\label{fig:oracle_vs_size}
\end{figure}

Two findings emerge. First, the oracle top-$k^{\star}$ on MuLocBench rises with repository size, indicating that multi-file fixes in larger codebases benefit from more candidate neighbors per anchor. A fixed default is a reasonable starting point but leaves headroom for a size-aware policy. Second, the oracle confidence threshold $\theta^{\star}$ on LocBench rises with repository size: in larger, denser repositories, low-confidence edges introduce more noise and stricter filtering helps. The fixed default $\theta=0.5$ thus over-filters small repos and under-filters parts of the multi-file regime, motivating learned or adaptive hyperparameter policies as future work.

\section{Implementation Details}
\label{sec:implementation-details}

\subsection{Graph Construction}
\label{app:graph-impl}

We instantiate $G_c=(V_c,E_c)$ as a typed multigraph over directories, files, classes, and functions, with edges drawn from a finite set of dependency types $\tau$. Construction proceeds in two parsing passes: a syntactic pass extracts intra-language relations (containment, import, invocation, inheritance), and a secondary pass adds cross-artifact links to tests, documentation, and configuration. After parsing, we attach two auxiliary signals consumed by the inner-loop scoring function (\hyperref[eq:scoring]{Eq.~\ref*{eq:scoring}}): an edge confidence $\omega$ and a file-level community label $\kappa$.

\paragraph{Edge confidence.}
To account for varying reliability of static analysis, we assign each semantic edge a fixed confidence score $\omega(e) = \omega_{\rho(e)} \in [0,1]$ indexed by the edge's provenance type $\rho(e)$ (\hyperref[tab:confidence]{Table~\ref*{tab:confidence}}); structural edges receive $\omega \equiv 1$. The threshold $\theta$ in \hyperref[eq:topk]{Eq.~\ref*{eq:topk}} filters out edges with $\omega(v,u) < \theta$, ensuring that local expansion prioritizes high-quality structural evidence under the agent's context budget.

\paragraph{Community partitioning.}
Issues that touch one functional region of a codebase typically require coordinated changes across files in that region, so we equip $G_c$ with a coarse cluster signal. We project the typed multigraph onto an undirected file-level graph $\widetilde G$ on file nodes $V_c^{\mathrm{file}}$, with edge weight $\widetilde w(\{u',v'\})$ counting cross-file semantic edges $\bigl\{(u,v) \in E_c^{\mathrm{sem}} : \mathrm{file}(u)=u',\,\mathrm{file}(v)=v',\,u'\neq v'\bigr\}$, and partition $\widetilde G$ into communities $\kappa : V_c^{\mathrm{file}} \to \mathbb{N}$ via the Leiden algorithm~\citep{traag2019louvain}. Each file additionally carries a cohesion score equal to the edge density of its induced subgraph. \hyperref[eq:scoring]{Eq.~\ref*{eq:scoring}} consumes $\kappa$ as a soft prior, biasing selection toward neighbors that share the anchor's cluster.

\subsection{Edge Confidence Hierarchy}
\label{app:confidence}
\hyperref[tab:confidence]{Table~\ref*{tab:confidence}} gives the full confidence hierarchy used by the scoring and gating mechanisms described in \hyperref[app:graph-impl]{Appendix~\ref*{app:graph-impl}}.

\begin{table}[h]
\centering
\small
\begin{tabular}{lclc}
\toprule
\textbf{Edge provenance} & $\omega$ & \textbf{Edge provenance} & $\omega$ \\
\midrule
Same-file co-occurrence & 1.0 & Cython implementation & 0.85 \\
Explicit import statement & 0.95 & Test linkage (\texttt{tested\_by}) & 0.75 \\
Resolved import & 0.9 & Documentation (\texttt{documents}) & 0.6 \\
Inheritance & 0.9 & Configuration (\texttt{configures}) & 0.5 \\
Fuzzy name match & 0.5 & & \\
\bottomrule
\end{tabular}
\caption{Edge confidence hierarchy. Scores reflect the reliability of the static-analysis evidence that produced each edge type.}
\label{tab:confidence}
\end{table}

\subsection{Sidecar Graph Storage}
Instead of serving graph neighborhoods through a live graph database, we materialize graph data as lightweight per-file JSON sidecar files. Each sidecar contains:
\begin{itemize}[left=0pt,nosep]
    \item Typed neighbors: dependents, dependencies, callers, callees, with edge types and confidence scores.
    \item Execution flows: process chains passing through the file.
    \item Community: Leiden community label and cohesion score.
    \item Cross-role links: related tests, documentation, and configuration files.
\end{itemize}

Neighborhoods are capped at 20 neighbors per file (10 in a compact variant) to bound sidecar size. This design keeps runtime lookup lightweight: neighborhood retrieval reduces to loading a small number of local JSON files, avoiding a graph database or MCP server into the retrieval loop.

\subsection{Prompts}
\label{app:prompts}

To isolate the contribution of graph evidence from prompt-engineering effects, the LARGER agent and the non-graph CLI baseline share an identical task framing, output schema, and step budget; the only difference is four lines of guidance that tell the agent how to interpret the structural signals appended to grep output.

The non-graph CLI agent runs reported in \hyperref[tab:main_results]{Table~\ref*{tab:main_results}} use the following baseline system prompt:

\begin{tcolorbox}[colback=gray!5,colframe=gray!50,title=Baseline agent system prompt,fontupper=\small\ttfamily]
You are a file localization agent. Given a GitHub issue (title and body), your task is to explore the repository and identify which source files need to be modified to resolve the issue.\\
\\
Instructions:\\
~~1. Read the issue carefully to understand the bug or feature request.\\
~~2. Explore the repository structure using shell commands (find, ls, grep).\\
~~3. Search for relevant code: function names, class names, error messages, or keywords mentioned in the issue.\\
~~4. Narrow down to the specific files that would need changes.\\
~~5. Output your answer as a JSON object with exactly this format:\\
~~~~~~~~\{"files\_to\_modify": ["path/to/file1.py", "path/to/file2.py"]\}\\
\\
Rules:\\
~~- You MUST list between 1 and 10 files. Aim for around 5 files in most cases.\\
~~- Think broadly: include source, test, documentation, and config files that would need changes.\\
~~- List files in order of importance; use paths relative to repo root.\\
~~- Only include files that actually exist in the repository.\\
~~- The final message MUST contain the JSON object and nothing else.
\end{tcolorbox}

The LARGER agent system prompt is identical to the baseline above, with Step 3 replaced by the following block (everything else is byte-identical):

\begin{tcolorbox}[colback=gray!5,colframe=gray!50,title=LARGER agent system prompt (Step~3 replacement),fontupper=\small\ttfamily]
~~3. Search for relevant code: function names, class names, error messages, or keywords mentioned in the issue. When grep finds matches, it also shows structurally related files from the dependency graph. Pay attention to these sections in grep output:\\
~~~~~~- [Related files from dependency graph] -{}- callers, callees, confidence scores\\
~~~~~~- Cluster labels group files by functional area\\
~~~~~~- Callers/Callees [c] -{}- confidence c in [0.5, 1.0] indicates how reliably the relationship was resolved (1.0 definite, 0.5 fuzzy)\\
~~~~~~- Flow annotations show execution paths; upstream/downstream files in a flow may also need changes\\
~~4. Use the graph context strategically:\\
~~~~~~- High-confidence callers/callees are strong candidates for related changes\\
~~~~~~- Files in the same cluster often need coordinated modifications\\
~~~~~~- Flows trace user-facing APIs to internal implementation\\
~~5. Narrow down to the specific files that would need changes.
\end{tcolorbox}

\paragraph{Graph evidence format.} The block appended to each grep result is generated deterministically from the sidecar index (\hyperref[app:graph-impl]{Section~\ref*{app:graph-impl}}); no LLM is used to produce it. A representative excerpt:

\begin{tcolorbox}[colback=gray!5,colframe=gray!50,title=Graph evidence format,fontupper=\small\ttfamily]
[Related files from dependency graph]\\
~~src/flask/blueprints.py (cluster: RoutingBlueprints):\\
~~~~Callers:~~src/flask/app.py:register\_blueprint [1.0],\\
~~~~~~~~~~~~~~src/flask/scaffold.py:\_endpoint\_from\_view\_func [0.9]\\
~~~~Callees:~~src/flask/helpers.py:url\_for [0.95]\\
~~~~Flow:~~~~~request\_dispatch (step 2/4)
\end{tcolorbox}

\paragraph{Downstream task evaluation.} For SWE-Atlas Codebase QA and Test Writing, we use the rubric-grading and execution-verification prompts distributed with the upstream Scale-AI \texttt{swe-atlas-qna} and \texttt{swe-atlas-tw} task releases without modification, so judging is identical across all methods compared in \hyperref[tab:sweatlas]{Table~\ref*{tab:sweatlas}}, and any prompt-induced variance is shared. The agent-side prompt for these tasks is the per-task instruction shipped with each Harbor task directory; LARGER and the CLI baselines receive byte-identical instructions and only differ in whether the grep channel is augmented with graph evidence.

\end{document}